\documentclass[11pt]{amsart}

\usepackage{amsfonts}
\usepackage{amsmath}
\usepackage{amssymb}
\usepackage{graphicx}%
\usepackage{comment}
\usepackage{epsf}
\usepackage{epsfig}
\usepackage{xypic}
 \theoremstyle{plain}

\newtheorem{corollary}{Corollary}[section]

\newtheorem{definition}{Definition}[section]
\newtheorem{example}{Example}[section]

\newtheorem{lemma}{Lemma}[section]

\newtheorem{proposition}{Proposition}[section]

\newtheorem{theorem}{Theorem}[section]

\numberwithin{equation}{section}

\makeatletter
\newcommand\twoheaddownarrow{%
\mathrel{\mathchoice
  {\raise2pt\hbox{%
  \ooalign{\hss$\downarrow$\hss\cr\lower2pt\hbox{%
  $\downarrow$}}}}
  {\raise2pt\hbox{%
  \ooalign{\hss$\downarrow$\hss\cr\lower2pt\hbox{%
  $\downarrow$}}}}
  {\raise1.5pt\hbox{%
  \ooalign{\hss$\scriptstyle\downarrow$\hss\cr\lower1.5pt\hbox{%
  $\scriptstyle\downarrow$}}}}
  {\raise1.1pt\hbox{%
  \ooalign{\hss$\scriptscriptstyle\downarrow$\hss\cr\lower1.1pt\hbox{%
  $\scriptscriptstyle\downarrow$}}}}
}}

\begin{document}

\vspace{0.5in}

\title[The Framework, Causal and Co-Compact Structure of Space-time]
{The Framework, Causal and Co-Compact  \\ Structure of Space-time}
\author{Martin Kov\'{a}r, Alena Chernikava}

\date{November 13, 2013}
\subjclass[2000]{} \keywords{Causal site, co-compact topology, formal context, Lorentzian manifold,
quantum gravity.}

\begin{abstract} We introduce a canonical, compact topology, which we call {\it weakly causal},
naturally generated by the causal site of J. D. Christensen and L. Crane, a pointless algebraic
structure motivated by certain problems of quantum gravity.  We show that for every
four-dimensional globally hyperbolic Lorentzian manifold there exists an associated causal site,
whose weakly causal topology is co-compact with respect to the manifold topology and vice versa.
Thus, the causal site has the full information about the topology of space-time, represented by the
Lorentzian manifold. In addition, we show that there exist also non-Lorentzian causal sites (whose
causal relation is not a continuous poset) and so the weakly causal topology and its de Groot dual
extends the usual manifold topology of space-time beyond topologies generated by the traditional,
smooth model. As a source of inspiration in topologizing the studied causal structures, we use some
methods and constructions of general topology and formal concept analysis.
\end{abstract}

%\thanks{}

\maketitle

\renewcommand\theenumi{\roman{enumi}}
\renewcommand\theenumii{\arabic{enumii}}
\font\eurb=eurb9 \font\seurb=eurb7
\def\integral{\int}
\def\cl{\operatorname{cl}}
\def\iff{if and only if }
\def\sup{\operatorname{sup}}
\def\clt{\operatorname{cl}_\theta}
\def\cli#1{\operatorname{cl}_{#1}}
\def\inti#1{\operatorname{int}_{#1}}
\def\Int{\operatorname{int}}
\def\intt{\operatorname{int}_\theta}
\def\id#1{\operatorname{\text{\sl id}}_{#1}}
\def\fr{\operatorname{fr}}
\def\ord{\operatorname{ord}}
\def\SIGMA{\operatorname{\Sigma}}
\def\cf{\operatorname{cf}}
\def\d{\operatorname{d}}
\def\diag{\operatorname{\Delta}}
\def\Mezera{\vskip 15 mm}
\def\mezera{\bigskip}
\def\Mezerka{\medskip}
\def\mezerka{\smallskip}
\def\iff{if and only if }
\def\G{\frak G}
\def\A{\Bbb A}
\def\I{\Bbb I}
\def\C{\mathcal C}
\def\F{\mathcal F}
\def\L{\mathcal L}
\def\P{\mathcal P}
\def\B{\frak B}
\def\O{\mathcal O}
\def\T{\frak T}
\def\X{\frak X}
\def\S{\mathcal S}
\def\K{\mathcal K}
\def\Q{\mathcal Q}
\def\R{\Bbb R}
\def\N{\Bbb N}
\def\Z{\Bbb Z}
\def\D{\Bbb D}
\def\Ds{\mathcal D}
\def\T{\mathcal T}
\def\zero{\bold 0}
\def\m{\frak m}
\def\n{\frak n}
\def\ts{  space }
\def\nbd{neighborhood }
\def\nbds{neighborhoods }
\def\card{cardinal }
\def\implies{\Rightarrow }
\def\map{\rightarrow }
\def\ekv{\Leftrightarrow }
\def\gre{\succcurlyeq}
\def\ngre{\not\succcurlyeq}
\def\gr{\succ}
\def\lre{\preccurlyeq}
\def\Ty{$T_{3.5}$ }
\def\Tyk{$T_{3.5}$}
\def\Slc{\text{\eurb Slc}}
\def\Top{\text{\eurb Top}}
\def\top{\text{\seurb Top}}
\def\Comp{\text{\eurb Comp}}
\def\TReg{\text{\eurb $\Theta$-Reg}}
\def\eK{\text{\eurb K}}
\def\M{\Bbb M}
\def\up{\uparrow\!\!}
\def\down{\downarrow\!\!}
\def\meet{\wedge}
\def\join{\vee}
\def\causeq{\preccurlyeq}
\def\caus{\prec}
\def\doubleprec{\prec\!\!\prec}

\section{Introduction}

\medskip

The belief that the causal structure of space-time is one  of its  most fundamental underlying
structures is almost as old as the idea of the relativistic space-time itself, but how is it
related to the topology of space-time? By tradition, there is no doubt regarding the topology of
space-time at least locally. It is usually considered to be locally homeomorphic with the Cartesian
power of the real line, equipped with the Euclidean topology. More recently, however, there
appeared ideas of discrete and pointless models of space-time in which the causal structure is
introduced axiomatically and so independently on the locally Euclidean models. Is, in these cases,
the axiomatic causal structure rich enough to carry also the full topological information? In
addition, after all, how the topology that we perceive around us and which is essentially and
implicitly at the background of many physical phenomena, may arise? The perceived topology of the
universe belongs to our reality similarly  as light, matter or gravitation, but the process of its
generation, certainly a fascinating, uninvestigated phenomenon, unfortunately has been ignored for
a long time.

The previous observation about $\R$ can be easily extended to any locally compact Hausdorff
topological space $(X, \tau)$, covering the most of the usual relativistic models of space-time. If
the original topology $\tau$ is compact Hausdorff, then both topologies coincide, so we can
restrict our further consideration to the non-compact case only. Then its counterpart -- the dual
topology $\tau^G$, is compact, T$_1$, superconnected, and what is also important, it is strictly
weaker than $\tau$ but it coincides with $\tau$ on every compact subset -- physically speaking, at
finite distances. Explaining the more complicated objects and phenomena by the simpler ones, is a
widely accepted and fundamental principle of science. In this sense, the co-compact topology
$\tau^G$ may be regarded as more fundamental for the natural world than the Euclidean topology of
Minkowski space or than the usual topology  $\tau$ of Lorentzian manifolds, since although it is
simpler, it still contains the full information about the usual topology in its compact subsets.

Since for a locally compact Hausdorff space we have $\tau=\tau^{GG}$, the preference of $\tau^G$
over $\tau$  seems to be rather a matter of the topologist's taste, or a modification of the
well-known dilemma {\it ``which came first -- chicken or the egg"}. However, there is an indication
that the compact T$_1$ topologies could be more fundamental for explanation of the physical
phenomenon of ``generating the topology" than the locally compact Hausdorff topologies. The reason
lies in the properties of the structure called a {\it causal site}, introduced by J. D. Christensen
and L. Crane in \cite{CC}. Causal site is a very general, pointless algebraic structure, capturing
the behavior of the light cones in relativity,  motivated by certain problems and developments in
quantum gravity. It is a natural generalization of R. Sorkin's causal sets, based on a modification
of the Grothendieck topology \cite{Ar}. In the next part of the paper we will find out that every
causal site generates, in a very natural way, a compact T$_1$ topology (and, in a contrast, not a
non-compact locally compact Hausdorff topology).

However, we will show even more than that -- the causal site can be selected in  such a way that
the generated topology will be the co-compact topology of the topology of any four-dimensional
globally hyperbolic Lorentz\-ian manifold. Considering this construction, it should be noted that
usually it is not a problem to choose certain regions of space-time as the topological subbase for
closed or open sets and reconstruct from them the de Groot dual of the original topology or even
the original topology itself, as well as the underlying set of points. The difficult part consists
in selecting appropriately these regions (with the corresponding binary relations) and proving,
that they satisfy all the axioms of a compatible causal site.

Nevertheless, the main purpose and sense of the algebraic structure of causal sites do not consist
in another, perhaps modified formulation of general relativity. The structure was designed in order
to formalize some alternative, very general (and perhaps hypothetical) models of space-time, very
different from the models that one can meet in classical general relativity. Our construction is
then a contribution in a topologization of these models of the generalized space-time, lying beyond
of the traditional, smooth model.

As a tool for topologization of various models of reality we will introduce a general construction,
suitable for equipping a set of objects with a topology-like structure, using the inner, natural
and intuitive relationships between them.

\bigskip

\section{Topological Prerequisites}\label{prerequisites}

Throughout this paper, we mostly use the usual terminology of general topology, for which the
reader is referred to \cite{Cs} or \cite{En}, with one exception -- in a consensus with a modern
approach to general topology, we no longer assume the Hausdorff separation axiom as a part of the
definition of compactness. This is especially affected by some recent motivations from computer
science, but also the contents of the paper \cite{HPS} confirms that such a modification of the
definition of compactness is a relevant idea. Thus, we say that a topological space is {\it
compact}, if every open cover of the space has a finite subcover, or equivalently, if every
centered system of closed sets or a closed filter base has a non-empty intersection. Note that by
the well-known Alexander's subbase lemma, the general closed sets may be replaced by more special
elements of any closed subbase for the topology.

We have already  mentioned the co-compact or the de Groot dual topology, which was first
systematically studied probably at the end of the 60's by  J. de Groot and his coworkers, J. M.
Aarts, H. Herrlich, G. E. Strecker and E. Wattel. The initial paper is \cite{Gro}. About 20 years
later, motivated by research in domain theory, the co-compact topology again came to the center of
interest of some topologists and theoretical computer scientists. During discussions in the
community the original definition due to de Groot was slightly changed to its current form,
inserting a word ``saturated" to the original definition (a set is saturated, if it is an
intersection of open sets, so in a T$_1$ space, all sets are saturated). Let $(X,\tau)$ be a
topological space. The topology, generated by the family of all compact saturated sets, used as the
base for the closed sets, we denote by $\tau^G$ and call it {\it co-compact} or {\it de Groot} dual
with respect to the original topology $\tau$. In \cite{LM} J. Lawson and M. Mislove stated
question, whether the sequence, containing the iterated duals of the original topology, is infinite
or the process of taking duals terminates after finitely many steps with topologies that are dual
to each other. In 2001 the first author solved the question and proved that only four topologies
may arise (see \cite{Kov1}).

The following theorem summarizes the previously mentioned facts important for understanding the
main results, contained in Section~\ref{causal}. The theorem itself is not new, under slightly
different terminology the reader can essentially find it in \cite{Gro}. A more general result,
equivalently characterizing the topologies satisfying $\tau=\tau^{GG}$, the reader may find in the
second author's paper \cite{Kov3}.  The proof we present here only for the reader's convenience,
without any claims of originality. For the proof we need to use the following notion. Let $\psi$ be
a family of sets. We say that $\psi$ has the finite intersection property, or briefly, that $\psi$
has {\it f.i.p.}, if for every $P_1, P_2,\dots, P_k\in\psi$ it follows $P_1\cap P_2\cap \dots\cap
P_k\ne\varnothing$. In some literature (for example, in \cite{Cs}), a collection $\psi$ with this
property is called {\it centered}.

\bigskip

\begin{theorem}\label{degroot} Let $(X,\tau)$ be a non-compact, locally compact Hausdorff topological space. Then
\begin{enumerate}
\item $\tau^{G}\subseteq \tau$, \item $\tau=\tau^{GG}$, \item $(X,\tau^{G})$ is compact and
superconnected, \item the topologies induced from $\tau$ and $\tau^G$ coincide on every compact
subset of $(X,\tau)$.
\end{enumerate}
\end{theorem}

\begin{proof} The topology $\tau^G$ has a closed base  which consists of  compact sets.
Since in a Hausdorff space all compact sets are closed, we have (i).

Let $C\subseteq X$ be a closed set with respect to $\tau$, to show that $C$ is compact with respect
to $\tau^G$, let us take a non-empty family $\Phi$ of compact subsets of $(X,\tau)$, such that the
family $\{C\}\cup\Phi$ has f.i.p. Take some $K\in\Phi$. Then the family $\{C\cap K\}\cup\{C\cap
F|\, F\in \Phi\}$ also has f.i.p. in a compact set $K$, so it has a non-empty intersection. Hence,
also the intersection of $\{C\}\cup\Phi$ is non-empty, which means that $C$ is compact with respect
to $\tau^G$. Consequently, $C$ is closed in $(X,\tau^{GG})$, which means that $\tau\subseteq
\tau^{GG}$. The topology $\tau^{GG}$ has a closed base consisting of sets which are compact in
$(X,\tau^{G})$. Take such a set, say $H\subseteq X$. Let $x\in X\smallsetminus H$. Since $(X,\tau)$
is locally compact and Hausdorff, for every $y\in H$ there exist $U_y, V_y\in \tau$ such that $x\in
U_y$, $y\in V_y$ and $U\cap V=\varnothing$, with $\cl U_y$ compact. Denote $W_y=X\smallsetminus\cl
U_y$. We have $y\in V_y\subseteq W_y$, so the sets $W_y$, $y\in H$ cover $H$. The complement of
$W_y$ is compact with respect to $\tau$, so $W_y\in \tau^G$. The family $\{W_y|\,y\in H\}$ is an
open cover of the compact set $H$ in $(X,\tau^G)$, so it has a finite subcover, say
$\{W_{y_1},W_{y_2},\dots,W_{y_k}\}$. Denote $U=\bigcap_{i=1}^k U_{x_i}$. Then $U\cap
H=\varnothing$, $x\in U\subseteq X\smallsetminus H$, which means that $X\smallsetminus H\in\tau$
and $H$ is closed in $(X,\tau)$. Hence, $\tau^{GG}\subseteq \tau$, an together with the previously
proved converse inclusion, it gives (ii).

Let us show (iii). Take any collection $\Psi$ of compact subsets of $(X,\tau)$ having f.i.p. They
are both compact and closed in $(X,\tau)$, so $\bigcap\Psi\ne\varnothing$. Then $(X,\tau^G)$ is
compact. Let $U,V\in\tau^G$ and suppose that $U\cap V=\varnothing$. The complements of $U$, $V$ are
compact in $(X,\tau)$ as intersections of compact closed sets in a Hausdorff space. Then $(X,\tau)$
is compact as a union of two compact sets, which is not possible. Hence, it holds (iii).

Finally, take a compact subset $K$ and a closed subset $C$ of $(X,\tau)$. Then $K\cap C$ is compact
in $(X,\tau)$ and hence closed  in $(X,\tau^G)$. Thus, the topology on $K$ induced from $\tau^G$ is
finer than the topology induced from $\tau$. Together with (i), we get (iv).
\end{proof}

\bigskip

\section{Causal Sites and Their Properties}\label{causal}

\medskip

In this section we will study the axiomatized causality relationships. Philosophically, the section
is inspired by the work of R. Sorkin and his notion of {\it causal set}, \cite{So}. However, we
will use a slightly different notion, due to J. D. Christensen and L. Crane, motivated by certain
problems in quantum gravity.

\medskip

Recall that a {\it causal site} $(\S,\sqsubseteq, \prec)$ defined by J. D. Christensen and L. Crane
in \cite{CC} is a set $\S$ of {\it regions} equipped with two binary relations $\sqsubseteq$,
$\prec$, where $(\S,\sqsubseteq)$ is a partial order having the binary suprema $\sqcup$ and the
least element $\bot\in \S$, and $(\S\smallsetminus\{\bot\},\prec)$ is a strict partial order (i.e.,
anti-reflexive and transitive), linked together by the following axioms, which are satisfied for
all regions $a, b, c\in \S$:
\begin{enumerate}

\item $b\sqsubseteq a$ and $a\prec c$ implies $b\prec c$,

\item $b\sqsubseteq a$ and $c\prec a$ implies $c\prec b$,

\item $a\prec c$ and $b\prec c$ implies $a\sqcup b\prec c$.

\item There exits $b_a\in \S$, called {\it cutting of $a$ by $b$}, such that {
     \begin{enumerate}

     \item $b_a\prec a$ and $b_a\sqsubseteq b$;

     \item if $c\in \S$, $c\prec a$ and $c\sqsubseteq b$ then $c\sqsubseteq b_a$.

     \end{enumerate}

}
\end{enumerate}

\medskip

\begin{lemma}\label{Linear} Let $(\S,\sqsubseteq, \prec)$ be a causal site with $L\subseteq \S\smallsetminus\{\bot\}$
linearly ordered by the relation $\sqsubseteq$, $\left|L\right|\geq 2$. Then $L$ is an anti-chain
with respect to $\prec$.
\end{lemma}

\begin{proof}  Let $a,b\in L\smallsetminus\{\bot\}$, $a\ne b$. Without loss of generality, we may choose the denotation of the
elements $a,b$ such that $b\sqsubseteq a$. Suppose for a moment, that $a\prec b$. Then by the axiom
(i) it should be $b\prec b$, Similarly, if $b\prec a$, by the axiom (ii) we get again that $b\prec
b$. But this is impossible, since $\prec$ is anti-reflexive on $\S\smallsetminus\{\bot\}$ by
definition. Hence, the subset $L\subseteq \S\smallsetminus\{\bot\}$ is an anti-chain with respect
to $\prec$.
\end{proof}

\begin{lemma}\label{Weakest} Let $(\S,\sqsubseteq, \prec)$ be a causal site. Then $\bot$ is the least element of $S$ with
respect to $\prec$, and $\bot\prec\bot$.
\end{lemma}

\begin{proof}
Let $a\in S$ be an arbitrary element. By the axiom (iv), there exists $a_a\in \S$ such that
$a_a\prec a$ and $a_a\sqsubseteq a$. By Lemma~\ref{Linear}, this is not possible if $a_a\in
\S\smallsetminus\{\bot\}$ since then $L=\{a_a, a\}$ would be a two-element linearly ordered subset
of $\S\smallsetminus\{\bot\}$. Hence $a_a=\bot$ and so $\bot\prec a$.
\end{proof}

\begin{corollary}\label{CausalSite1} Let $(\S,\sqsubseteq)$ be a linearly ordered set with the least element
$\bot$. Then there exists a unique binary relation $\prec$ on $\S$, such that $(\S,\sqsubseteq,
\prec)$ is a causal site. In this relation, $\S\smallsetminus\{\bot\}$ is an anti-chain and $\bot$
is the least element of $(\S,\prec)$.
\end{corollary}

It is not difficult to show, that for every inclusion relation $\sqsubseteq$, the corresponding
causal relation $\prec$ always exists. One can simply define $\prec$ on $\S\smallsetminus\{\bot\}$
as anti-chain and take $\bot$ as the least element of $(\S,\prec)$ similarly as in the previous
example. It follows from Lemma~\ref{Weakest} that such a causal relation is the weakest possible
one; on the other hand, if one takes the causal relation $\prec$ as the primary relation, an
appropriate inclusion relation may not exist.

\begin{corollary}\label{CausalSite2} Let $(\S, \prec)$ be a strictly linearly ordered set, $\left|\S\right|\geq 3$.
Then there is no binary relation $\sqsubseteq$ such that $(\S,\sqsubseteq, \prec)$ is a causal
site.
\end{corollary}

\begin{proof} Suppose conversely, that $(\S,\sqsubseteq, \prec)$ is a causal site with $\prec$
linear, $\left|\S\right|\geq 3$. There exist $b,c\in \S\smallsetminus\{\bot\}$, $b\ne c$. Since
$(\S,\sqsubseteq)$ has binary suprema by the definition of a causal site, there exist $a=b\sqcup
c\in \S\smallsetminus\{\bot\}$. Clearly, $b, c\sqsubseteq a$. Since $b\ne c$, at least one of the
elements $b, c$ differs from $a$. Let $L$ be the two-element set containing $a$ and that element
from $\{b,c\}$, different from $a$. Then $L$ is a chain with respect to $\sqsubseteq$ as  well as
$\prec$, which is not possible by Lemma~\ref{Linear}.
\end{proof}

\begin{theorem} Let $(\S,\sqsubseteq, \prec)$ be a causal site and let
$N(x)=\{y|\, y\in \S, y\nprec x\nprec y\}$. Then for every distinct $x_1, x_2,\dots x_k\in
\S\smallsetminus\{\bot\}$, where $k\ge 2$,
$$\bigcap_{i=1}^k N(x_i)\ne\varnothing.$$
\end{theorem}

\begin{proof} Let $x_1, x_2,\dots x_k\in \S$ and denote $s=x_1\sqcup x_2 \sqcup \dots \sqcup x_k$.
It is clear that $x_i\sqsubseteq s$ for every $i=1,2,\dots, k$. Suppose, for some $i=1,2,\dots, k$,
that it holds $x_i\prec s$ or $s\prec x_i$. Because of anti-reflexivity of the relation $\prec$ on
$\S\smallsetminus\{\bot\}$, the equality $x_i=s$ is possible only for $s=\bot$, which immediately
gives that $x_1=x_2=\dots=x_k=\bot$. However, this contradicts to the assumptions of the theorem.
Hence, $x_i\ne s$. Then $L=\{x_i,s\}\subseteq\S\smallsetminus\{\bot\}$, $\left|L\right|=2$ is
linearly ordered by $\sqsubseteq$, so by Lemma~\ref{Linear} $L$ is an anti-chain with respect to
$\prec$. But this is a contradiction to our assumption that $x_i\prec s$ or $s\prec x_i$. Hence,
$s\in\bigcap_{i=1}^k N(x_i)\ne\varnothing.$

\end{proof}

The previous theorem gives a number of examples illustrating, how a causal site does not look like.

\begin{corollary}
Let $M\subseteq\S\smallsetminus\{\bot\}$ be a finite set such that for each $x\in S$, there exists
some $m\in M$ with $x\prec m $ or $m\prec x$. Then there is no relation $\sqsubseteq$ such that
$(S,\sqsubseteq, \prec)$ is a causal site.
\end{corollary}

Because of anti-reflexivity of $\prec$, there may exist finite causal sites in spite of the
corollary (Corollary~\ref{CausalSite1} yields such a construction). However, the limitation is very
strong, as we will see later.

If not specified otherwise, let $\preccurlyeq$ be the denotation of the reflexive closure of the
binary relation $\prec$. Recall that an atom in a poset with the least element $\bot$ is a minimal
element among all elements different from $\bot$ (that is, atoms are the immediate successors of
the least element). Then, for instance, $(\S,\preccurlyeq)$ cannot be a finite lattice for
$\left|\S\right|>2$ or even any other poset with the top element, which is not an atom. It also
cannot be an infinite `ladder' whose Hasse diagram is on the Figure~1, since one can let, for
instance, $M=\{2,3\}$. There are also many other possibilities, which we leave to the reader as an
easy exercise.

\medskip

{\tiny
$$\diagram
&   \\
& 5 \udashed \\
4\udashed \urline & 3 \uline \\
2 \uline \urline & 1 \uline \\
\bot \uline \urline &
\enddiagram$$}
{\medskip \centerline{\scriptsize\rm Figure 1.} } \medskip

The interesting and very complex problem of characterization of those posets $(\S,\preccurlyeq)$
for which the causal site $(\S,\sqsubseteq, \prec)$ exists for some appropriate order relation
$\sqsubseteq$ (and, perhaps, have also some physical counterpart in the reality) exceeds the aim of
this paper. However, as we will see later,  for every four-dimensional globally hyperbolic
Lorentzian manifold there exists an appropriate causal site generated by its causal structure. So
it is natural to ask whether there exist also some `non-Lorentzian' causal sites -- those, which
cannot be properly associated with a Lorentzian manifold or its subset. In the following
proposition we will describe a construction which yields such examples. For any set $S$, by $S^F$
we denote the family of all finite subsets of~$S$.

\begin{proposition}\label{CausalSiteFromPoset} Let $(P,\le)$ be any poset. For every $K, L\in P^F$ we put
$$K\prec L \text{ if and only if }   k<l \text{ for every } (k,l)\in K\times L.$$
Then $(P^F, \subseteq, \prec)$ is a causal site.
\end{proposition}

\begin{proof} It is obvious that the conditions (i) -- (iii) from the definition of the causal site are
satisfied. We will verify (iv). Let $A,B \in P^F$. We put
$$B_A=\{b|\, b\in B, b<a \text{ for every } a\in A\}.$$
It is clear that $B_A\prec A$, $B_A\subseteq B$. Take any $C\in P^F$, such that $C\prec A$,
$C\subseteq B$. Let $c\in C$. Then $c\in B$, $c<a$ for all $a\in A$, so $C\subseteq B_A$.
\end{proof}

Varying the poset $(P,\le)$, one can generate causal sites which cannot be Lorentzian. The easiest
way it is to take $P$ with cardinality greater than it can have any connected Lorentzian manifold.
We will construct another, a more suitable example. At first, let us recall some notions. Let
$(P,\le)$ be a poset. We say that $x$ is {\it way below} $y$, and write $x\ll y$ for $x,y\in P$, if
for every directed set $D\subseteq P$ with a supremum $\sup D\in P$, the relation $y\le\sup D$
implies that there exists some $d\in D$ with $x\le d$. We say that the poset $(P,\le)$ is {\it
continuous}, if for every $x\in P$, the set $\twoheaddownarrow\{x\}=\{t|\, t\in P, t\ll x\}$  is
directed and $x=\sup\twoheaddownarrow\{x\}$. For more detail, the reader is referred to
\cite{GHKLMS}. A continuous poset $(P,\le)$ is called {\it bicontinuous}, if it is continuous also
with respect to the inverse order relation $\ge$, and the corresponding way-below relations are
inverse to each other. For a more precise definition, the reader is referred to \cite{MP1}.

\begin{corollary}\label{NonContinuous} Let $(P, \le)$ be a poset with the Hasse diagram as on the Figure~2, in which two infinitely
countable chains $\{1,3,\dots\}$ and $\{2,4,\dots\}$ have their common supremum $\omega$ and the
infimum $\bot$. Then $(P^F,\subseteq,\prec)$ is a causal site in which $(P^F, \preccurlyeq)$ is not
a continuous poset.
\end{corollary}

\medskip

{\tiny

$$\diagram
& \omega & \\
  \urdashed &  &  \uldashed \\
5 \udashed &  & 6 \udashed \\
3 \uline &  & 4 \uline \\
1 \uline &  & 2 \uline \\
& \bot \ulline  \urline &
\enddiagram$$}
{\medskip \centerline{\rm\scriptsize Figure 2.}} \medskip

\begin{proof}
Let $(P, \le)$ be a poset defined by the Hasse diagram on the Figure~2.
Proposition~\ref{CausalSiteFromPoset} ensures that $(P^F,\subseteq,\prec)$ is a causal site. It
remains to show that $(P^F, \preccurlyeq)$ is not a continuous poset. Indeed, $\{\omega\}$ is an
upper bound of the chains $\mathcal L_1=\{\{1\},\{3\},\dots\}$ and $\mathcal
L_2=\{\{2\},\{4\},\dots\}$. Let $D\in P^F$ be an an upper bound for $L_i$ for $i=1$ or $i=2$
respectively. Then for every $k=1,2,\dots$ and $x\in D$ it follows $2k-1<x$ or $2k<x$,
respectively. In the both cases, this is possible only for $x=\omega$, so $\sup\mathcal
L_1=\{\omega\}=\sup\mathcal L_2$.

Now, let $K\doubleprec\{\omega\}$ for some $K\in P^F$. By the definition of the way-below relation,
there exist $m,n\in\N$ such that $K\preccurlyeq\{2m-1\}$,  $K\preccurlyeq\{2n\}$. Then $K$ cannot
contain any element of $\N$, since otherwise it would be odd and even at the same time. Hence,
$K=\{\bot\}$. Then $\{\omega\}$ is clearly not the supremum of $\twoheaddownarrow\{\omega\}$, so
the poset $(P^F, \preccurlyeq)$, as well as $(P, \le)$, is not continuous.
\end{proof}

Corollary~\ref{NonContinuous} shows, that there exist causal sites which are not covered by the
theory of P. Panangaden and K. Martin, who studied causality via domain theory (see \cite{MP1} and
\cite{MP2}). In their approach, the causal relation is a bicontinuous poset and the corresponding
counterparts of regions are the intervals with respect to the way-below relation. Since P.
Panangaden and K. Martin proved that every globally hyperbolic Lorentzian manifold can be
represented in this way, Corollary~\ref{NonContinuous} also illustrates that there exist also
`non-Lorentzian'  causal sites. On the other hand, there is no evidence that the `regions'
constructed from the way-below intervals satisfy the axioms of a causal site.

\bigskip

The previous corollaries illustrate that in a causal site the `inclusion' relation $\sqsubseteq$
and the `causal' relation $\prec$ highly depend one on each other. Once a causal site is given, its
corresponding topology is fully determined by the `inclusion' relation $\sqsubseteq$, but as it is
shown by Corollary~\ref{CausalSite2}, for a given causal relation $\prec$, the relation
$\sqsubseteq$ cannot be arbitrary. In Corollary~\ref{CausalSite2}, the relation $\prec$ is too
strong to allow the existence of an appropriate inclusion relation $\sqsubseteq$.

These facts yield a reply for a certain criticism, complaining that the topology derived from
$\sqsubseteq$ completely ignores causality and so it is unphysical. This is certainly not true,
since for a given, existing causal site $(\S, \sqsubseteq, \prec)$, the relations $\sqsubseteq$ and
$\prec$ are not independent, and a properly chosen causal site has a real, physical sense. On the
contrary, our approach leaves open the possibility that also the other, uninvestigated topologies,
compatible with the causal relation, could be potentially useful.

\bigskip

\begin{example} {\rm

Consider a set $Q$ of particles. An interaction between two particles $p, q\in Q$ may be
interpreted as an element $(p,q)$ of the Cartesian product $Q\times Q$. But since it is natural to
admit that various particles may interact several times, we rather represent their interaction as
the triple $(p,q,n)\in Q\times Q\times \N$.

Suppose that for every particle $p\in Q$ there exists a poset $(Q(p),\le_p)$ such that
$Q(p)\subseteq Q\times Q\times \N$ and $\pi_1(Q(p))=p$, where $\pi_1:Q\times Q\times \N\map Q $ is
the projection to the first component of the product $Q\times Q\times \N$.

The poset $(Q(p),\le_p)$ represents the chronological order of the interactions of the particle $p$
with other  particles from $Q$. Note that $(Q(p),\le_p)$ need not linear if one admits possible
existence of alternative histories of particle interactions. Let $(P,\le)$ be the disjoint sum of
the posets $(Q(p),\le_p)$, $p\in Q$. By Proposition~\ref{CausalSiteFromPoset},
$(P^F,\subseteq,\prec)$ forms a causal site, which, in some cases, need not be Lorentzian. \hfill
$\square$

}
\end{example}

\bigskip

\section{How to Topologize Everything}\label{howto}

\medskip

As it has been recently noted in \cite{HPS},  nature or  physical universe, whatever it is, has
probably no existing, real points like in the classical Euclidean geometry (or, at least, we cannot
be completely sure of that). Points, as a useful mathematical abstraction, are infinitesimally
small and thus cannot be measured or detected by any physical way. However, what we can be sure
that really exists, there are various locations, containing concrete physical objects. In this
paper we will call these locations {\it places}. Various places can overlap, they can be merged,
embedded or glued together, so the theoretically understood virtual ``observer" can visit multiple
places simultaneously. For instance, the Galaxy, the Solar system, the Earth, (the territory of)
Europe, Brno (a beautiful city in the Czech Republic, the place of the first author's home),
Mogilev (a charming and scenic town in Belarus, the second author's home), the room in which the
reader is present just now, is a simple and natural example of places conceived in our sense.
Certainly, in this sense, one can be present at many of these places at the same time, and, also
certainly, there exist pairs of places (e.g. Brno and Mogilev), where the simultaneous presence of
any physical objects is not possible. Thus, the presence of various physical objects connects these
primarily free objects -- our places -- to the certain structure, which we call a {\it framework}.

Note that it does not matter that the places are, at the first sight, determined rather vaguely or
with some uncertainty. They are conceived as elements of some algebraic structure, without any
additional geometrical or metric structure and as we will see later, the ``uncertainty" could be
partially eliminated by the relationships between them. Let's now give the precise definition.

\begin{definition} Let $\P$ be a set, $\pi\subseteq 2^\P$. We say that $(\P,\pi)$ is a
framework. The elements of $\P$ we call {\it places}, the set $\pi$ we call {\it framology}.
\end{definition}

Although every topological space is a framework by the definition, the elementary interpretation of
a framework is very different from the usual interpretation of a topological space. The elements of
the framology are not primarily considered as neighborhoods of places, although it seems to be also
very natural. If $\P$ contains all the places that are or can be observed, the framology $\pi$
contains the list of observations of the fact that the virtual ``observer" or some physical object
that ``really exists" (whatever it means), can be present at some places simultaneously. The
structure which $(\P,\pi)$ represents arises from these observations.

\medskip

Let us introduce some other useful notions.

\begin{definition} Let $(\P,\pi)$ and $(\S,\sigma)$ be frameworks. A mapping $f:\P\map \S$ satisfying
$f(\pi)\subseteq\sigma$ we call a {\it framework morphism}.
\end{definition}

\begin{definition} Let $(\P,\pi)$ be a framework, $\sim$ an equivalence relation on $\P$. Let
$\P_\sim$ be the set of all equivalence classes and $g:\P\map \P_\sim$ the corresponding quotient
map. Then $(\P_\sim, g(\pi))$ is called the quotient framework of $(\P,\pi)$ (with respect to the
equivalence $\sim$).
\end{definition}

\begin{definition}A framework $(\P,\pi)$ is T$_0$ if for every $x,y\in \P$, $x\ne y$, there
exists $U\in\pi$ such that $x\in U$, $y\notin U$ or $x\notin U$, $y\in U$.
\end{definition}

\begin{definition} Let $(\P,\pi)$ be a framework. Denote $\P^d=\pi$ and $\pi^d=\{\pi(x)|\, x\in \P\}$,
where $\pi(x)=\{U|\, U\in\pi, x\in U\}$. Then $(\P^d, \pi^d)$ is the {\it dual} framework of
$(\P,\pi)$. The places of the dual framework $(\P^d, \pi^d)$ we call {\it abstract points} or
simply {\it points} of the original framework $(\P,\pi)$.
\end{definition}

The framework duality is a simple but handy tool for switching between the classical point-set
representation (like in topological spaces) and the point-less representation, introduced above.

\bigskip

{\bf Some Examples.} There are a number of  natural examples of mathematical structures satisfying
the definition of a framework, including non-oriented graphs, topological spaces (with open maps as
morphisms), measurable spaces or texture spaces  of M. Diker \cite{Di}. Among physically motivated
examples, we may mention the Feynman diagrams with particles in the role of places, and
interactions as the associated abstract points. Very likely, certain aspects of the string theory,
related to general topology, can also be formulated in terms of the framework theory.

\bigskip

It should be noted that the notion of a framework is a special case of the notion of the {\it
formal context}, due to B. Ganter and R. Wille \cite{GW}, sometimes also referred as the Chu space
\cite{ChL}. Recall that a formal context is a triple $(G,M, I)$, where $G$ is a set of objects, $M$
is a set of attributes and $I\subseteq G\times M$ is a binary relation. Thus, the framework
$(\P,\pi)$ may be represented as a formal context $(\P,\pi, \in)$, where objects are the places and
their attributes are the abstract points.

Even though the theory and methods of formal concept analysis may be a useful tool also for our
purposes, we prefer the topology-related terminology that we introduced in this section because it
seems to be more close to the way, how mathematicians and physicists usually understand to the
notion of  space-time. Nevertheless,  we completely share the opinion of R. Wille expressed in
\cite{Wi}, that the abstract mathematical disciplines, like the lattice theory and other, should be
returned more close to their potential users. We hope that this section will contribute to the same
purpose regarding general topology.

It also seems that frameworks are closely related to the notion of partial metric due to S.
Matthews \cite{Ma}, but these relationships will be studied in a separate paper.

\begin{proposition}Let $(\P,\pi)$ be a framework. Then $(\P^d,\pi^d)$ is T$_0$.
\end{proposition}

\begin{proof} Denote $\S=\pi$, $\sigma=\{\pi(x) |\, x\in \P\}$, so $(\S,\sigma)$ is the dual framework
of $(\P, \pi)$. Let $u, v\in \S$, $u\ne v$. Since $u, v\in 2^\P$ are different sets, either there
exists $x\in u$ such that $x\notin v$, or there exists $x\in v$, such that $x\notin u$. Then
$u\in\pi(x)$ and $v\notin\pi(x)$, or $v\in\pi(x)$ and $u\notin\pi(x)$. In both cases there exists
$\pi(x)\in\sigma$, containing one element of $\{u, v\}$ and not containing the other.
\end{proof}

\begin{theorem} Let $(\P,\pi)$ be a framework. Then $(\P^{dd}, \pi^{dd})$ is isomorphic to the
quotient of $(\P,\pi)$. Moreover, if $(\P,\pi)$ is T$_0$, then $(\P^{dd}, \pi^{dd})$ and $(\P,\pi)$
are isomorphic.
\end{theorem}

\begin{proof} We denote $\mathcal R=\P^d=\pi$, $\rho=\pi^d=\{\pi(x) |\, x\in \P \}$,
$\S=\mathcal R^d=\rho$, $\sigma=\rho^d=\{\rho(x) |\, x\in \mathcal R \}$. Then $(\S, \sigma)$ is
the double dual of $(\P,\pi)$. It remains to show, that $(\S, \sigma)$ is isomorphic to some
quotient of $(\P,\pi)$.

For every $x\in \P$, we put $f(x)=\pi(x)$. Then $f:\P\map \S$ is a surjective mapping. It is easy
to show, that $f$ is a morphism. Indeed, if $U\in\pi$, then $f(U)=\{\pi(x) |\, x\in U\}=\{\pi(x)
|\, x\in \P,  U\in\pi(x) \}=\{V |\, V\in\rho,  U\in V\}=\rho(U)\in\sigma$. Therefore,
$f(\pi)\subseteq\sigma$, which means that $f$ is an epimorphism of the framework $(\P,\pi)$ onto
$(\S,\sigma)$.

Now, we define $x\sim y$ for every $x, y\in \P$ \iff $f(x)=f(y)$. Then $\sim $ is an equivalence
relation on $\P$. For every equivalence class $[x]\in \P_\sim$ we put $h([x])=f(x)$. The mapping
$h: \P_\sim\map \S$ is correctly defined, moreover, it is a bijection. The verification that $h$ is
a framework isomorphism is standard, but, because of completeness, it has its natural place here.
The quotient framology on $\P_\sim$ is $g(\pi)$, where $g:\P\map \P_\sim$ is the quotient map. The
quotient map $g$ satisfies the condition $h\circ g=f$. Let $W\in g(\pi)$. There exists $U\in \pi$
such that $W=g(U)$. Then $h(W)=h(g(U))=f(U)\in\sigma$. Hence $h(g(\pi))\subseteq\sigma$, which
means that $h: \P_\sim\map \S$ is a framework morphism. Conversely, let $W\in\sigma=\{\rho(U) |\,
U\in\pi\}$. We will show that $h^{-1}(W)\in g(\pi)$. By the previous paragraph, $\rho(U)=f(U)$ for
every $U\in\pi$, so there exists $U\in\pi$, such that $W=f(U)=h(g(U))$. Since $h$ is a bijection,
it follows that $h^{-1}(W)=g(U)\in g(\pi)$. Hence, also $h^{-1}:\S \map \P_\sim$ is a framework
morphism, so the frameworks $(\P_\sim, g(\pi) )$ and $(\S,\sigma)$ are isomorphic.

Now let us consider the special case when $(\P,\pi)$ is T$_0$. Suppose that $f(x)=f(y)$ for some
$x,y\in \P$. Then $\pi(x)=\pi(y)$, which is possible only when $x=y$. Then the relation $\sim$ is
the diagonal relation, and the quotient mapping $g$ is an isomorphism.
\end{proof}

\begin{corollary} Every framework arise as dual \iff it is T$_0$.
\end{corollary}

\begin{corollary} For every framework $(\P,\pi)$, it holds  $(\P^d,\pi^d)\cong (\P^{ddd},
\pi^{ddd})$.
\end{corollary}

\medskip

Note, that the framework structure is suitable for addressing the compatibility  problem of various
scales in physics and their different models. Since the points of the Universe probably do not
exist in reality (although they are a useful mathematical abstraction), the abstract points of a
framework only express certain relationships between places, which -- in a contrast to points --
can be really observed and which exclusively exist in the physical reality. Then various
framologies and various topological models may peacefully coexist with help of the framework
duality on a given set $P$ of places.

\begin{definition} Let $(\P,\pi)$ be a framework, $(X,\tau)$ be a topological space with the family
$\C$ of closed sets. We say that $(X,\tau)$ is an {\it open} ({\it closed}, respectively) {\it
topological model} for $(\P,\pi)$, if there exist a framework $(\S,\sigma)$ isomorphic to
$(\P,\pi)$ and set $X^\prime\subseteq X$ such that $\S\subseteq \tau$ ($\S\subseteq \C$,
respectively) and $\sigma=\{\{A|\, A\in\S, x\in A\}|\, x\in X^\prime\}$.
\end{definition}

\begin{example} {\rm Let $p_i$, $i\in \Z$ be pairwise distinct elements. We put $\P=\{p_i|\, i\in\Z\}$,
$\pi=\{\{p_i,p_{i+1}\}|\,i\in\Z\}$. Then the framework $(\P, \pi)$  has many open as well as closed
topological models, including the real line $\R$ equipped with the Euclidean topology and the
Khalimsky line, that is, the set $\Z$ with the topology generated by its open base
$\tau_K=\{\dots,\{1\},\{1,2,3\},$ $\{3\},\{3,4,5\},\{5\},\dots\}$. For an easy proof, one can the
places $p_i$ identify with non-empty, open (or closed, respectively) overlapping sets in a
topological space, such that $p_i$ has a non-empty intersection only with $p_{i-1}$, $p_i$ and
$p_{i+1}$. In addition, in case of Khalimsky topology, one can simulate various scales by taking
$p_i$ with more or less elements, although the original framework $(\P, \pi)$ still remains the
same.

}
\end{example}

\bigskip

\section{Topologies from Causal Sites}\label{causal}

\medskip

In this section we show that the notion of a~framework, introduced and studied in the previous
section, has some real utility and sense. In a contrast to simple examples mentioned above, from a
properly defined  framework we will be able to construct a topological structure with a real,
physical meaning.

\medskip

Consider a causal site $(P,\sqsubseteq, \prec)$ and let us define appropriate framework structure
on $P$. We say that a subset $F\subseteq P$ set is centered, if for every $x_1, x_2, \dots, x_k\in
F$ there exists $y\in P$, $y\ne\bot$ satisfying $y\sqsubseteq x_i$ for every $i=1,2,\dots, k$. If
$\L\subseteq 2^P$ is a chain of centered subsets of $P$ linearly ordered by the set inclusion
$\subseteq$, then $\bigcup \L$ is also a centered set. Then every centered $F\subseteq P$ is
contained in some maximal centered $M\subseteq P$. Let $\pi$ be the family of all maximal centered
subsets of $P$. Now, consider the framework $(P,\pi)$ and its dual $(P^d, \pi^d)$. Let $(X,\tau)$
be the topological space with $X=P^d=\pi$ and the topology $\tau$ generated by its closed subbase
(that is, a subbase for the closed sets) $\pi^d$.

\begin{theorem}\label{comp} The topological space $(X,\tau)$, corresponding to the framework $(P^d,\pi^d)$ and the causal
site $(P,\sqsubseteq, \prec)$, is compact T$_1$.
\end{theorem}

\begin{proof} By the well-known Alexander's subbase lemma, for proving the compactness of $(X,\tau)$ it is sufficient
to show, that any subfamily of $\pi^d$ having the f.i.p., has nonempty intersection. The subbase
for the closed sets of $(X,\tau)$ has the form $\pi^d=\{\pi(x)|\, x\in P\}$, so any subfamily of
$\pi^d$ can be indexed by a subset of $P$. Let $F\subseteq P$ and suppose that for every $x_1,
x_2,\dots, x_k\in F$ we have
$$\pi(x_1)\cap\pi(x_2)\cap\dots\cap\pi(x_k)\ne\varnothing.$$ Then there exists $U\in\pi$ such that
$U\in\pi(x_1)\cap\pi(x_2)\cap\dots\cap\pi(x_k)$, so $x_i\in U$ for every $i=1,2,\dots, k$. Since
$U$ is a (maximal) centered family, there exists $\bot\ne y\in P$ such that $y\sqsubseteq x_i$ for
every $i=1,2,\dots,k$. Thus, $F$ is a centered family, contained in some maximal centered family
$M\subseteq P$. Then we have $M\in\pi$, so $$M\in\bigcap_{x\in M}\pi(x)\subseteq\bigcap_{x\in
F}\pi(x)\ne\varnothing.$$ Hence, $(X,\tau)$ is compact.

Let $U,V\in X=\pi$, $U\ne V$. Since both are maximal centered subfamilies of $P$, none of them can
contain the other one. So, there exist $x, y\in P$ such that $x\in U\smallsetminus V$ and $y\in
V\smallsetminus U$. Then $U\in\pi(x)$, $V\notin\pi(x)$, $V\in\pi(y)$, $U\notin\pi(y)$. Thus,
$X\smallsetminus\pi(x)$, $X\smallsetminus\pi(y)$ are open sets in $(X,\tau)$ containing just one of
the points $U, V$. So the topological space $(X,\tau)$ satisfies the T$_1$ axiom.
\end{proof}

\medskip

Note that the topology, generated by a causal site by the way described in this section, we call
the {\it weakly causal topology}.

\bigskip

\section{Causal Sites from Lorentzian Manifolds}\label{Lorentz}

\medskip

In this section we will construct a causal site from an arbitrary, four-dimensional Lorentzian
manifold which does not admit of a closed non-spacelike curve. We will also show, that the weak
topology, generated from this causal site, is the de Groot dual of the manifold topology.

\medskip

For most notions and terminology used in this section we refer the reader to the book \cite{HE} by
S. W. Hawking and G. F. R. Ellis. Throughout this section, if not specified otherwise, by $M$ we
will denote a connected, four-dimensional, Hausdorff, $C^\infty$-differentiable manifold, equipped
with a Lorentz metric $g:M\times M\map \R$, a metric tensor of signature $+2$ on $M$. In addition,
we will suppose that $M$ satisfies the causality condition, that is, there are no closed
non-spacelike curves in $M$, or in other words, that $M$ is globally hyperbolic.

\medskip

The following two lemmas we will use as our starting point. They summarize some results which one
can find in \cite{HE} (especially in Section 4.5).

\begin{lemma}\label{local curves} Let $U\subseteq M$ be a convex normal coordinate neighborhood.
The following statements hold:

\begin{enumerate}

\item Every $p\in U$ can be reached from some $q\in U$ by a timelike, future-oriented curve and
from some $r\in U$ by a timelike, past-oriented curve.

\item If $p\in U$ can be reached from $q\in U$ by a non-spacelike curve but not by a timelike
curve, then $p$ lies on a null geodesic from $q$.
\end{enumerate}
\end{lemma}

For the proof, the reader is referred to \cite{HE}.

\bigskip

\begin{lemma}\label{global curves} Let $p,q\in M$, $p\ne q$. Then the following statements are
fulfilled:

\begin{enumerate}
\item If $p$ and $q$  are joined by a non-spacelike curve $\gamma$, which is not a null geodesics,
they can also be joined by a timelike curve, say $\zeta$. \item Moreover, if $\gamma$ lies totally
in a normal convex coordinate neighborhood $U$, then $\zeta$ also can be chosen in such a way that
it totally lies in $U$.
\end{enumerate}
\end{lemma}
Note that for the proof, the reader is referred to \cite{HE}, where only (i) is explicitly stated.
The proof technique used, however, allows showing also the part (ii).

\medskip

\begin{lemma}\label{convexclosure} Let $V\subseteq M$ be convex. Then also $\cl V$ is convex.
\end{lemma}

\begin{proof} Suppose for a moment that there exists a geodesic $\zeta:\left[0,1\right]\map M$ connecting
the points $p_0,q_0\in\cl V$ which does not lie totally in $\cl V$. The geodesics in $M$ depend
differentiably and hence continuously on their initial and terminal points. Therefore, there exists
a continuous function $f:\left[0,1\right] \times U\times W$, where $U$, $W$ are certain open
neighborhoods of the points $p_0$, $q_0$, such that $f(t,p_0,q_0)=\zeta(t)$, and $f(t,p,q)$ is the
geodesic connecting the general points $p\in U$, $q\in W$, defined on the interval
$\left[0,1\right]$. It means that $f(0,p,q)=p$ and $f(1,p,q)=q$.

By our initial assumption, there exists $t_0\in\left[0,1\right]$ such that
$f(t_0,p_0,q_0)=\zeta(t_0)\in M\smallsetminus \cl V$. From continuity of $f$ it follows that there
exist open $U_0\subseteq U$, $W_0\subseteq W$ such that $p_0\in U_0$, $q_0\in W_0$ and for every
$(p,q)\in U_0\times W_0$ it holds $f(t_0,p,q)\in M\smallsetminus \cl V$. However,  since
$p_0,q_0\in\cl V$, there exist $p_1\in U_0\cap V\ne\varnothing$, $q_1\in W_0\cap V\ne\varnothing$.
There is only one geodesic connecting the points $p_1$, $q_1$ -- $f(t,p_1,q_1)$, and because of
convexity of $V$, whole $f(t,p_1,q_1)$ is  contained in $V$. This contradicts to our previous
conclusion $f(t_0,p,q)\in M\smallsetminus \cl V$, following from the continuity of $f$.
\end{proof}

\bigskip

\begin{lemma}\label{boundarycross} Let $p\in M$ and let $\varphi:U\map \R^4$ be a normal coordinate system at $p$ with the coordinate functions $x^i$,
$\varphi(p)=0$. Then there exists arbitrarily small $\varepsilon>0$ such that the open set
$V_\varepsilon=\{q|\, q \in U,\sum_{i=1}^4(x^i(q))^2<\varepsilon\}$ has the following property: If
$\zeta:\left[0,1\right]\map M$ is a geodesic emanating from $r\in V_\varepsilon$ with
$\zeta(t_0)\in\fr V_\varepsilon$, then for every open neighborhood $W$ of $\zeta(t_0)$ there exists
some $t\in \left[0,t_0\right]$ such that $\zeta(t)\in V_\varepsilon\cap W$.
\end{lemma}

\begin{proof} Let us choose $\varepsilon>0$ sufficiently small such that $V_\varepsilon$ is convex and the symmetric $(0,2)$-tensor having the
components $\delta_{ij}-\Gamma_{ij}^k x^k$ (where $\Gamma_{ij}^k$ are the Christoffel symbols) is
positively definite on $\cl V_\varepsilon$. In fact, the second assumption is standardly used for
the proof of convexity of $V_\varepsilon$ (for example, in \cite{ON}). However, we will slightly
modify the standard technique for our purposes.

Suppose conversely that there exists an open neighborhood $W$ of $\zeta(t_0)$ such that
$\zeta(\left[0,t_0\right])\cap V_\varepsilon\cap W=\varnothing$. Without loss of generality we may
assume that $W$ is convex. Since $\cl V_\varepsilon$ is convex by Lemma~\ref{convexclosure},
$\zeta(\left[0,t_0\right])\subseteq \cl V_\varepsilon$, which means that $\zeta(\left[0,t_0\right])
\cap W\subseteq\fr V_\varepsilon$. By continuity of $\zeta$, $\zeta^{-1}(W)$ is an open
neighborhood of $t_0$ in the topological space $\left[0,1\right]$, and so there exists some
$w\in\zeta^{-1}(W)\cap\left[0,t_0\right)$. Then $\zeta(w)\in \zeta(\left[0,t_0\right])\cap W$ and
since $W$ is convex, it follows $\zeta(\left[w, t_0\right])\subseteq \zeta(\left[0,t_0\right])\cap
W\subseteq\fr V_\varepsilon$. Consider the function
\begin{equation}
f(t)=\sum_{i=1}^4 (x^i\circ\zeta)^2(t).
\end{equation}
For every $t\in \left[w, t_0\right]$ it holds $f(t)=\varepsilon^2$. On the other hand,
\begin{equation}
{\d^2 f\over \d t^2}=2\sum_{i=1}^4\left(\left({{\d (x^i\circ\zeta)}\over{\d
t}}\right)^2+(x^i\circ\zeta)\cdot{{\d^2(x^i\circ\zeta)}\over{\d t^2}}\right),
\end{equation}
from which using the geodesic equation for $\zeta$
\begin{equation}
{{\d^2(x^i\circ\zeta)}\over{\d t^2}}+\Gamma_{jk}^i {{\d (x^j\circ\zeta)}\over{\d t}}\cdot {{\d
(x^k\circ\zeta)}\over{\d t}}=0
\end{equation}
it is easy to get
\begin{equation}
{\d^2 f\over \d t^2}=2\left(\delta_{ij}-\Gamma_{ij}^k x^k\right){{\d (x^i\circ\zeta)}\over{\d
t}}\cdot {{\d (x^j\circ\zeta)}\over{\d t}}>0,
\end{equation}
since the tensor $\delta_{ij}-\Gamma_{ij}^k x^k$ is positively definite on $\cl V_\varepsilon$. But
this is a contradiction with the constant value $\varepsilon^2$ of the function $f$ on the interval
$\left[w, t_0\right]$.
\end{proof}

\bigskip

\begin{lemma}\label{convexintersection} For every $p\in M$ there exists an open local base $\sigma_p\subseteq \tau$  at $p$, such that every $V\in\sigma_p$ has the
following properties:
\begin{enumerate}
\item $V$ is convex. \item $K=\cl V$ is compact. \item There exists a normal convex coordinate
neighborhood $U$ such that $K\subseteq U$. \item For every $r\in\Int K$, $K\cap J^-(r)$ is regular
closed.
\end{enumerate}
\end{lemma}

\begin{proof} Let us choose a normal coordinate system $\varphi:U\map \R^4$ at $p$. The set $\varphi(U)$ is open in the
Euclidean topology of $\R^4$, so $\varphi(p)$ has an open neighborhood, say $O$, such that $\cl O
\subseteq\varphi(U)$ and $\cl O$ is compact. Hence, $S=\varphi^{-1}(O)$ is open in $M$,
$H=\varphi^{-1}(\cl O)$ is compact and hence closed in the Hausdorff topological space $M$, and
$S\subseteq H$.  We put $\sigma_p= \{V_\varepsilon|\, \varepsilon>0, V_\varepsilon\subseteq S\}$,
where $V_\varepsilon$ is the convex open neighborhood of $p$ whose existence is ensured by
Lemma~\ref{boundarycross}. Since $V_\varepsilon\subseteq S\subseteq H\subseteq U$, $\cl
V_\varepsilon\subseteq H$ is compact and contained in $U$. It is also clear that $\sigma_p$ is a
local base at $p$, so the conditions (i), (ii) and (iii) are satisfied.

To show (iv), let us take some fixed $V\in\sigma_p$ and denote $K=\cl V$. We will show that $K\cap
J^-(r)=\cl (\Int (K\cap J^-(r)))$. Consider the set $U\cap J^-(r)$. It is shown in \cite{HE} that
the boundary of $J^-(r)$ in $U$ is formed by null geodesics, so $U\cap J^-(r)$ is a closed
topological subspace of $U$. Hence, there exists a closed set $H\subseteq M$ such that $U\cap
J^-(r)=U\cap H$. Then, $K\cap J^-(r)=K\cap U\cap  J^-(r)=K\cap U\cap H=K\cap H$, which is a closed
set in $M$. Therefore, $\cl (\Int (K\cap J^-(r)))\subseteq K\cap J^-(r)$.

Conversely, take any $x\in  K\cap J^-(r)$ and select an open neighborhood, say $W$, of $x$. We will
show that $W$ meets $\Int(K\cap J^-(r))$. In any case, $x$ can be an inner or a boundary point of
$J^-(r)$. If $x\in \Int J^-(r)$, then $W\cap \Int J^-(r)$ is also an open neighborhood of $x$ which
must meet $V$ since $x\in K=\cl V$. Then $\varnothing\ne W\cap \Int J^-(r)\cap V \subseteq W\cap
\Int J^-(r)\cap \, \Int K= W\cap (\Int(K\cap J^-(r)))$, so $W$ meets $\Int(K\cap J^-(r))$. Now,
suppose the other possibility, that $x\in \fr J^-(r)$. Then $x$ lies on a null geodesic, say
$\zeta$, emanating  from $r$. From Lemma~\ref{boundarycross} it follows that there exists
$x^\prime\in W\cap V\subseteq W\cap\Int K$, also lying on $\zeta$. However, $W\cap\,\Int K$ is an
open neighborhood of $x^\prime$, which must meet $\Int J^-(r)$ since $x$ is a boundary point of
$J^-(r)$. Then, $\varnothing\ne W\,\cap\Int K\cap\Int J^-(r)=W\cap\Int(K\cap J^-(r))$, so again $W$
meets $\Int(K\cap J^-(r))$. Hence $x\in \cl(\Int(K\cap J^-(r)))$. It follows $K\cap
J^-(r)\subseteq\cl(\Int(K\cap J^-(r)))$, which completes the proof.
\end{proof}

\bigskip

Let $x, y\in M$. We put $x\leqslant y$ if $x=y$ or there exists a non-spacelike future-oriented
curve $\psi:\left[0,1\right]\map M$ with $\psi(0)=x$, $\psi(1)=y$. It is clear that $\leqslant$ is
a reflexive and transitive relation, so it is a preorder on $M$. The relation need not necessarily
be antisymmetric.
\bigskip

\begin{definition}\label{multidiamond}
For every $p\in M$ we define $J^-(p)=\{x|, x\in M, x\leqslant p\}$, $J^+(p)=\{x|, x\in M,
p\leqslant x\}$. For  $F,G\subseteq M$ we put
$$F \Diamond G=\bigcap_{p\in F} J^+(p)\cap \bigcap_{q\in G} J^-(q).$$
We say that the set $F \Diamond G$ is a {\it multi-diamond} if the following conditions are
satisfied:

\begin{enumerate}
\item $F, G$ are non-empty and finite.
%\item $\Int (F \Diamond G)$ is nonempty.
\item $F \Diamond G$ is compact. \item There exists a  normal convex coordinate neighborhood
$U_{F,G}\subseteq M$ and an open set $V_{F,G}$ with $K_{F,G}=\cl V_{F,G}\subseteq U_{F,G}$ compact
 such that $F\cup G\subseteq K_{F,G}$, $F\Diamond G\subseteq\Int K_{F,G}$.
\item For every $r\in\Int K_{F,G}$, the set $K_{F,G}\cap J^-(r)$ is regular closed.
\end{enumerate}

\end{definition}
\bigskip

\bigskip
\medskip

Let $\mathcal D$ be a family of all such multi-diamonds $F \Diamond G$  (with non-empty interior)
and let $\P=\mathcal D^F$ be the family of all finite unions of elements of $\mathcal D$ (including
the union of the empty collection, so we admit $\varnothing\in\P$) The family $\mathcal D$ is
closed under finite intersections and similarly, $\P$ is closed under finite unions by its
definition. Let $A, B\in \P$. Then there exist multi-diamonds $C_1, C_2, \dots C_n \in \mathcal D$
and $D_1, D_2,\dots D_m\in \mathcal D$ such that $A=\bigcup_{i=1}^n C_i$ and $B=\bigcup_{j=1}^m
D_j$. Then
$$A\cap B=(\bigcup_{i=1}^n C_i)\cap(\bigcup_{j=1}^m D_j)=\bigcup_{i=1}^n\bigcup_{j=1}^m (C_i\cap
D_j),$$ and since $C_i\cap D_j\in\mathcal D$, we also have $A\cap B\in\P$. Hence, $\P$ is also
closed under finite intersections.

\medskip

\begin{lemma}\label{interiors} Let $W\subseteq M$ be an open set. Then for every
$p\in W$ there exists $A\in \P$ with $p\in\Int A$ and $A\subseteq W$.
\end{lemma}

\begin{proof} Take an open set $W\subseteq M$ and $p\in W$. By Lemma~\ref{convexintersection}, there exists a
convex open neighborhood $V\subseteq W$ of $p$ such that $K=\cl V$ is compact, contained in some
normal convex coordinate neighborhood $U$ and such that $K\cap J^-(r)$ is a regular closed set for
every $r\in\Int K$. It is also clear that there exist $u, v\in V$, $u\leqslant v$, with
$p\in\Int(\{u\}\Diamond\{v\})\subseteq \{u\}\Diamond\{v\}\subseteq V\subseteq W$. Furthermore, the
set $A=\{u\}\Diamond\{v\}$ is compact as a closed subspace of the compact set $K$. By
Definition~\ref{multidiamond}, $A$ is a multi-diamond, and so $A\in \P$.
\end{proof}

\begin{corollary} The family $\{\Int A|\, A\in \P\}$ is a base of the topology of $M$.
\end{corollary}

\begin{lemma}\label{cocompactbase} The family $\P$ is a closed base for the co-compact topology on $M$.
\end{lemma}

\begin{proof}
The co-compact topology $\tau_M^G$ on $M$ is generated by its open base, which is formed by the
complements of sets, compact in the manifold topology $\tau_M$. Let $K\subseteq M$ be compact.
Denote $U=M\smallsetminus K$.  Take a point $x\in U$. Since $M$ is a T$_1$ space, $M\smallsetminus
\{x\}$ is an open set with respect to the manifold topology. By Lemma~\ref{interiors}, for every
$y\in K$, there exists $A_y\in \P$ such that $y\in\Int A_y\subseteq A_y\subseteq M\smallsetminus
\{x\}$, which also means that $x\notin A_y$.

Since $K$ is compact, there exist $y_1, y_2, \dots, y_k\in K$ such that $$K\subseteq
\bigcup_{i=1}^k \Int A_y.$$ Then $$x\in \bigcap_{i=1}^k (M\smallsetminus A_{y_i})=M\smallsetminus\
\bigcup_{i=1}^k A_{y_i}\subseteq M\smallsetminus K=U,$$ and the closed set $\bigcup_{i=1}^k
A_{y_i}$ is an element of $\P$. Hence, every set $U$, which is open with respect to  $\tau_M^G$, is
a union of complements of elements of $\P$, which are closed in the same topology. Then $\P$ forms
a closed base for $\tau_M^G$.
\end{proof}

\medskip

\medskip

\begin{lemma}\label{maximize}
Let $A=F\Diamond G$ be a multi-diamond, $p\in M$, $A_p=(F\smallsetminus\{p\})\Diamond G$. Let
$K=K_{F,G}$ be the compact set introduced in Definition \ref{multidiamond}. Then for each $x\in
O_p$ where
$$O_p=\{x|\, x\in K, A_p\cap J^+(x)=A\}$$
there exists a maximal element $m\in O_p$ such that $x\leqslant m$.

\end{lemma}

\begin{proof} At first, let us show that the set $O_p$ is nonempty. If $p\in F$, then $A_p\cap J^+(p)=A$ and so $p\in O_p$.
On the other hand, if $p\notin F$, it follows $A_p=A$, so the condition $A_p\cap J^+(x)=A$ is
equivalent to $A\subseteq J^+(x)$. But then $\varnothing\ne F\subseteq O_p$.

Now, let $L\subseteq O_p\subseteq K$ be a non-empty linearly ordered chain with respect to
$\leqslant$. We will show that $L$ has an upper bound in $O_p$. Consider the net $\id
L(L,\leqslant)$. Since $K$ is compact $\id L(L,\leqslant)$ has a cluster point, say $p_L\in K$.
Suppose that there is some $l\in L$ such that $p_L\notin J^+(l)$. Since the set $K\cap J^+(l)$ is
closed, there exists an open neighborhood $U$ of $p_L$ with $U\cap K\cap J^+(l)=\varnothing$. By
the definition of the cluster point, there exists $m\in L$, $l\leqslant m$, such that $m\in U$.
Then $m\in U\cap K \cap J^+(m)$, but this is not possible since $J^+(m)\subseteq J^+(l)$. Hence,
$p_L\in \bigcap_{l\in L} J^+(l)$, which means that $p_L$ is an upper bound of $L$ in $K$.

It remains to show that $p_L\in O_p$, which is equivalent to verify that $A_p\cap J^+(p_L)=A$. Let
$l\in L\subseteq O_p\subseteq K$. Then $A_p\cap J^+(l)=A$ and since $l\leqslant p_L$, we have
$J^+(p_L)\subseteq J^+(l)$. We will show that $A\subseteq J^+(p_L)$. Suppose conversely, that there
exists some $r\in A\smallsetminus J^+(p_L)$. Since $K\cap J^+(p_L)$ is closed (and $M$ certainly is
a metrizable topological space) there exists $\varepsilon >0$ such that $B_\varepsilon(r)\cap K\cap
J^+(p_L)=\varnothing$. Since by Definition~\ref{multidiamond} $A\subseteq \Int K$, without loss of
generality we may select $\varepsilon>0$ in such a way that $B_\varepsilon(r)\subseteq K$.

Since $K$ is contained in a normal convex coordinate neighborhood by Definition \ref{multidiamond},
for every $x\in K$ there exist a unique geodesic  $\zeta_x$ emanating from $x$ and terminating at
$r$ and a unique geodesic, say $\eta_x$, connecting $x$ with $p_L$. Moreover, if $x\in K\cap \Int
J^-(r)$, the geodesic $\zeta_x$ must be timelike and future-oriented.

Denote by $f(x)$ the terminal point of the  geodesic $\zeta_x^\prime$ emanating from $p_L$ of the
same length as $\zeta_x$, with the corresponding tangent vector at $p_L$ parallely transported
along $\eta_x$ from $x$. Then $f(x)$ depends differentiably and hence continuously on $x$,
moreover, on the compact set $K$ the mapping  $f(x)$ is uniformly continuous by the Heine-Cantor
theorem. Hence, there exists some $\delta>0$ such that if for some $p\in K$, $x\in B_\delta(p)$,
then $f(x)\in B_\varepsilon(f(p))$. Since $p_L$ is a cluster point of the net $\id L(L,\leqslant)$,
there exists $z\in B_{\delta}(p_L)\cap L$. Since $z\in L$, it follows $r\in A\subseteq J^+(z)$,
which means that $z\in J^-(r)$. Hence $B_{\delta}(p_L)\cap K\cap J^-(r)\ne\varnothing$ and since
$K\cap J^-(r)$ is a regular closed set, there exists some $x\in B_{\delta}(p_L)\cap \, \Int (K\cap
J^-(r))= B_{\delta}(p_L)\cap \, \Int K\cap \, \Int J^-(r)$. Now, $x$ is in the domain of $f$ and by
its uniform continuity, $f(x)\in B_\varepsilon( f(p_L))$. From the construction of the function $f$
and the uniqueness of the geodesics $\zeta_x$ and $\zeta_x^\prime$ it follows $f(p_L)=r$. Since
$f(x)$ is the terminal point of the non-spacelike future-oriented curve $\zeta_x^\prime$, emanating
from $p_L$, it follows that $f(x)\in B_\varepsilon(r)\cap J^+(p_L)\subseteq K$,  which is a
contradiction. Hence, $A\subseteq J^+(p_L)$. Then
$$A=A\cap J^+(p_L)\subseteq A_p\cap J^+(l)=A.$$
and so $p_L\in O_p$ is the upper bound of the chain $L$. Let $M_p$ be the set of all maximal
elements of $O_p$ (with respect to $\leqslant $). By Zorn's Lemma, for every $x\in O_p$ there
exists $m\in M_A$ such that $x\leqslant m$.
\end{proof}

Since $p_L\in O_p$, every element of $F$ can be replaced by an element, which is maximal among
those, yielding the same set $F\Diamond G$.

\medskip

\begin{lemma}\label{cones1} Let $U\subseteq M$ be a normal convex neighborhood, $p,q\in U$ points that are
connected by a null geodesic. Then the following conditions hold:
\begin{enumerate}
\item $U\cap\Int J^-(p)\cap J^+(q)=\varnothing$, \item $U\cap J^-(p)\cap \Int J^+(q)=\varnothing$,
\item The set $U\cap\,\fr J^-(p)\cap \fr J^+(q)$ is equal to the closed geodesical segment between
the points $p$, $q$.
\end{enumerate}
\end{lemma}

\begin{proof} (i) Suppose conversely that there is some $z\in U\cap\Int J^-(p)\cap J^+(q)$. Then there exists a future-oriented, non-spacelike path connecting $q$
with $z$ and a future-oriented, timelike path  connecting $z$ with $p$, so there exists a
non-spacelike path, say $\zeta$, from $q$ to $p$, which lies totally in $U$. Since $\zeta$ is not a
null geodesic, $p\notin \fr J^+(q)$, so $p\in \Int J^+(q)$. Hence, there exists a timelike
geodesics joining the points $p$, $q$. Then $p$, $q$ are connected by two different geodesics,
which is not possible in the normal convex neighborhood $U$.

(ii) is analogous to (i).

(iii) Since $U$ is geodesically convex, it is clear that the geodesical segment $\zeta$, connecting
$p$, $q$ is totally contained in $U$ and every its point is a boundary point of $J^-(p)$, $J^+(q)$.
Conversely, let $x\in U\cap\fr J^-(p)\cap \fr J^+(q)$. Then, there is a future-oriented null
geodesic emanating from $q$ to $x$ and a future-oriented null geodesic connecting $x$ with $p$.
Suppose that $x$ does not lie on $\zeta$. Then there exists a future-oriented, non-spacelike  path,
say $\gamma$, from $q$ to $p$, totally lying in $U$, which is not a null geodesic. By
Lemma~\ref{global curves}, $\gamma$ can be varied to a timelike geodesic, connecting the points
$p$, $q$ in $U$. Then $p$, $q$ are connected by two different geodesics, which certainly is not
possible.

\end{proof}

\begin{lemma}\label{maxfinite}  Let $A=F\Diamond G\in \mathcal D$ be a nonempty
multi-diamond, $K=K_{F,G}$ be the compact set introduced in Definition \ref{multidiamond}. Then
there exists a finite set $M_A$ containing all maximal elements of
$$O_A=\{x|\, x\in K, A\subseteq J^+(x)\}.$$
\end{lemma}

\begin{proof} Suppose $A$ having the form $A=F \Diamond G=\bigcap_{p\in F} J^+(p)\cap \bigcap_{q\in G}
J^-(q)$, where $F$, $G$ are proper finite subsets of $M$. By Lemma~\ref{maximize}, without loss of
generality we may assume, that the set $F$ contains only maximal elements among those, giving the
same value for $A=F \Diamond G$.

Let $m$ be a maximal element of $O_A$. The cone $J^+(m)$ contains $A$, which is a closed compact
set under the assumptions stated in the lemma. Then also $\fr A\subseteq J^+(m)$. Suppose for a
moment, that $\fr A\subseteq A\subseteq \Int J^+(m)$. Then for every $x\in A$ there exists a
timelike future-oriented curve $\zeta_x:[0,1]\map M$ with $\zeta_x(0)=m$, $\zeta_x(1)=x$. Denote
$V_x=\Int J^+(\zeta_x({1\over 2}))$. Then $x\in V_x$, so $\Omega=\{V_x|\, x\in A$ is an open cover
of the compact set $A$. Hence, $\Omega$ has a finite subcover, say $\{V_{x_1}, V_{x_2},\dots
V_{x_k}\}$. The set $\bigcap_{i=1}^k \Int J^-(\zeta_{x_i}({1\over 2}))$ is an open neighborhood of
$m$, so by Lemma~\ref{local curves}, there exists an open set $U\subseteq M$ such that $m\in
U\subseteq \bigcap_{i=1}^k \Int J^-(\zeta_{x_i}({1\over 2}))$ and some $r\in U$ which can be
reached from $m$ by a timelike, future-oriented curve. Let $x\in A$. There exists some
$j\in\{1,2,\dots,k\}$ such that  $x\in V_{x_j}=J^+(\zeta_{x_j})$ and since also $r\in
J^-(\zeta_{x_j})$, there exists a timelike future-oriented curve which joins $r$ with $x$. It means
that $A\subseteq J^+(r)$, which contradicts to the maximality of $m$. Hence, $\fr J^+(m)$ meets
$\fr A$. On the other hand, $\fr J^+(m)$ cannot meet the interior of $A$ since the opposite leads
to a nonempty intersection of $\Int A$ with the complement of $J^+(m)$, which then contradicts to
$A\subseteq J^+(m)$. The boundary of $A$ can be decomposed into the union $\fr A=C_1\cup
C_2\cup\dots\cup C_n$, where each $C_i$ is a closed compact (and smooth) piece of the boundary
$E_i$ of $J^+(q)$ or $J^-(q)$ for some $q\in F\cup G$.
% Each $E_i$ has its own equation (pozor, platna lokalne na zvolenem souradnicovem okoli), which can
% be, for simplicity, identified with $E_i$. The relevant part of $\fr J^+(m)$ has also its own
% equation, say $E$.
Consider the set $C_i\cap \fr J^+(m)$ for some $i=1,2,\dots, n$, for which $C_i\cap\fr J^+(m)$ is
nonempty. It can contain only finitely many points or, otherwise, it constitutes a submanifold of
$M$ of the dimension equal to $1$, $2$ or $3$. We will distinguish several possibilities, and in
each of them we will show that $m\in F\cup G$.

\medskip

(i) At first, consider the most simple case, that $\dim (C_i\cap\fr J^+(m))=3$. Since both $\fr
J^+(m)$ and $C_i$ are submanifolds of $M$ of dimension $3$, on some non-empty open subset of $M$
their corresponding equations must coincide, so the cones $J^+(m)$, $J^\pm(q)$ are the same. But
then $m=q\in F\cup G$.

\medskip

Now, suppose that $\dim (C_i\cap \fr J^+(m))<3$. Then the submanifolds $\fr J^+(m)$ and $C_i$ can
either touch or cut one another.

\medskip

(ii) The first case means that $\dim (C_i\cap\fr J^+(m))=1$ and the corresponding three-dimensional
tangent vector spaces $T_p C_i$, $T_p \fr J^+(m)$ coincide at any point $p\in C_i\cap\fr J^+(m)$.
Take some fixed $p\in C_i\cap\fr J^+(m)$, $p\notin F\cup G\cup\{m\}$, and some coordinate system
$\varphi=(x^1, x^2, x^3, x^4)$, where we also denote $x^4=t$, defined on some open coordinate
neighborhood $U_p$ of $p$, such that $\varphi(p)=(x^1_p,x^2_p,x^3_p,t_p)$.
% There exists a smooth curve $\chi(s)$ parameterizing $C_i\cap\fr J^+(m)$ on some neighborhood of
% $p$ and restricting the domain of $\varphi$ if necessary we may arrange that this neighborhood is
% $U_p$. Further, without loss of generality we may assume that $\chi(s)$ is defined on some open
% interval $I$ containing $0$, and $\chi(0)=p$.

Recall that by $J^\pm(q)$, where $q\in F\cup G$ is some proper element, we previously denoted the
cone having the boundary $E_i=\fr J^\pm (q)$ containing $C_i$. Let $\sigma:M\times M:\map\R$ be the
Synge's world function (see \cite{Po}), defined by
\begin{equation}\label{Synge}
\sigma(x, x^\prime)={1\over 2}(\lambda_1-\lambda
_0)\integral_{\lambda_0}^{\lambda_1}(g_{ij}\circ\beta)(\lambda)\, {\partial (x^i\circ\beta)\over
\partial \lambda} {\partial (x^j\circ\beta)\over \partial \lambda}\  \d\lambda,
\end{equation}
where the integral is evaluated on the geodesic $\beta$ with the affine parameter $\lambda$,
linking $x$, $x^\prime$ such that $\beta(\lambda_0)=x^\prime$ and $\beta(\lambda_1)=x$. It is not
difficult to show that Synge's world function vanishes on a geodesic if and only if it is a null
geodesic, \cite{Po}. Then the equations of $C_i$, $\fr J^+(m)$ on $U_p$ are
\begin{equation}\label{Syngechi}
\sigma(x,q)=0, \quad \sigma(x,m)=0.
\end{equation}
The fact that $T_p C_i=T_p \fr J^+(m)$ means that the normal vectors to $T_p C_i$, $T_p \fr J^+(m)$
coincide, so
\begin{equation}\label{geo1}
{\partial\over\partial x^i}\sigma(\varphi^{-1}( (x^1(p),\dots,x^4(p)),q)={\partial\over\partial
x^i}\sigma(\varphi^{-1}( (x^1(p),\dots,x^4(p)),m).
\end{equation}
However, the quantities ${\partial(\sigma\circ\varphi^{-1})\over\partial x^i }$, apart from a
linear factor, express also the covariant components of the tangent vector of the geodesics
emanating from $z\in\{q,m\}$ to $p$. Hence, \ref{geo1} means that $p$, $q$, $m$ lie on the same
geodesic, say $\zeta$. From $p\in C_i$ it follows that $p$ is connected with $q$ by a null
geodesic, so by the standard theory of differential equations  it coincides with $\zeta$, which
hence is null. Therefore, $q\leqslant m$ or $m\leqslant q$.

(ii$_1$) Suppose that $E_i=\fr J^-(q)$. Since $\varnothing\ne A\subseteq J^-(q)\cap J^+(m)$, it
follows $m\leqslant q$. By the construction of $A$, there exists a normal convex neighborhood $U$,
containing $q, m$ and $A$. Suppose that there exist some $x\in A\cap \Int J^-(q)$. Then $U\cap\Int
J^-(q)$ is an open neighborhood of $x$. By (i) of Lemma \ref{cones1} it is not possible that $x\in
J^+(m)$, which is a contradiction with $A\subseteq J^+(m)$. Hence, $A\cap \Int J^-(q)=\varnothing$,
which means that $A\subseteq \fr J^-(q)$. By a similar way, from (ii) of Lemma~\ref{cones1} it
follows, that $A\subseteq \fr J^+(m)$, so $A\subseteq \fr J^-(q)\cap \fr J^+(m)$. By (iii) of
Lemma~\ref{cones1}, $A$ is a subset of the geodesical segment, connecting the points $q$, $m$.
Since $A$ is closed, from maximality of $m$ it follows $m\in A$, which implies that $m$ is an upper
bound of $F$. Then $J^+(m)\subseteq \bigcap_{a\in F} J^+(a)$, so
$$J^+(m)\cap\bigcap_{b\in G}J^-(b)\subseteq A=\bigcap_{a\in F} J^+(a)\cap \bigcap_{b\in G} J^-(b)\subseteq J^+(m)\cap\bigcap_{b\in G}J^-(b),$$
where the last inclusion follows from the fact that $A\subseteq J^+(m)$, and clearly $A\subseteq
\bigcap_{b\in G}J^-(b)$. Hence, $A=J^+(m)\cap\bigcap_{b\in G}J^-(b)$. By the assumption stated in
the first paragraph of this proof, $F=\left\{m\right\}$.

(ii$_2$) Now, suppose that $C_j\subseteq E_j=\fr J^+(q)$. Consequently, it is not possible that
$q\leqslant m$ and $q\ne m$ since $q$ would not be a maximal element among those giving the same
value for the set $A=F\Diamond G$. Indeed, then $A\subseteq J^+(m)\subseteq J^+(q)$, which implies
that $A=J^+(m)\cap A\subseteq J^+(m)\cap A_q\subseteq J^+(q)\cap A_q=A$, where we denoted
$A_q=(F\smallsetminus\{q\})\Diamond G$ as in Lemma~\ref{maximize}. Then $q$ could be replaced by
the greater element $m$, which is a contradiction. Hence, it holds $m\leqslant q$. Then $A\subseteq
J^+(q)\subseteq J^+(m)$ and from the assumption of maximality of $m$ it follows that $m=q\in F$.

\medskip

%\centerline{\includegraphics[width=10cm]{lorentz.eps}} \centerline{\small Figure 1.}

\bigskip

(iii) The second case means that the submanifolds $\fr J^+(m)$ and $C_i$ cut each other and $\dim
(C_i\cap\fr J^+(m))=2$. Let $p\in C_i\cap\fr J^+(m)$. The set $C_i\cap\fr J^+(m)$ is infinite, so
$p$ can be chosen in such a way that $p\notin F\cup G\cup\{m\}$.  We again select a coordinate
system $\varphi=(x^1, x^2, x^3, x^4)$, where we also denote $x^4=t$, defined on some open
coordinate neighborhood $U_p$ of $p$, such that $\varphi(p)=(x^1_p,x^2_p,x^3_p,t_p)$. There exists
a smooth function of two variables $\chi(u,v)$ parameterizing $C_i\cap\fr J^+(m)$ on some
neighborhood of $p$ and restricting the domain of $\varphi$ if necessary we may assume that this
neighborhood is $U_p$. Further, without loss of generality we may assume that $\chi(u,v)$ is
defined on an open set of the form $I\times I$, where $I$ is an open interval containing $0$, and
$\chi(0,0)=p$.

Since $\fr J^+(m)$ does not meet $\Int A$, the non-empty intersection $C_i\cap\fr J^+(m)$ must be
separated from the open set $U_p\cap\,\Int A\subseteq \Int A$ by another boundary piece $C_j$ on
some open neighborhood $V_p$ of $p$. Restricting the domains of $\varphi$ and $\chi$ if necessary
we may ensure that for simplicity $V_p=U_p$ and $\chi(I\times I)\subseteq C_i\cap C_j \cap\fr
J^+(m)$. Recall that in our previous notation, $C_i$, $C_i$ are a pieces of the boundaries $E_i$,
$E_j$ of some cones $J^\pm(q)$, $J^\pm(r)$, respectively, where $q,r\in F\cup G$. Hence, the points
$m,q,r$ lie on null geodesics emanating from the same point $\chi(u,v)$.

We will analyze the situation in the tangent vector space $T_p M$ of $M$. Let $\sigma:M\times
M:\map\R$ be the Synge's world function.  Then
\begin{equation}\label{Syngechi}
\sigma(\chi(u,v),z)=0
\end{equation}
for every fixed $z\in\{m,q,r\}$ and varying $(u,v)\in I\times I$. After differentiating
\ref{Syngechi} with respect to $u$ and $v$, using the chain rule we get
\begin{equation}\label{orto1}
{\partial(\sigma\circ\varphi^{-1})\over\partial x^i }\cdot{\partial (x^i\circ\chi)\over \partial
u}=0
\end{equation}
and
\begin{equation}\label{orto2}
{\partial(\sigma\circ\varphi^{-1})\over\partial x^i }\cdot{\partial (x^i\circ\chi)\over \partial
v}=0,
\end{equation}
where all derivatives are calculated at $p$ for each $z\in\{m,q,r\}$. The quantities
${\partial(\sigma\circ\varphi^{-1})\over\partial x^i}$, apart from a linear factor, are the
components of the rescaled tangent vector of the null geodesics emanating from $p$ to $z$ with
respect to the covariant basis in $T_p M$. On the other hand, ${\partial (x^i\circ\chi)\over
\partial u}$ and ${\partial (x^i\circ\chi)\over \partial u}$ are linearly independent generators of
the two-dimensional tangent vector space $T_p(C_i\cap C_j \cap\fr J^+(m))$.

As it follows from \ref{orto1} and \ref{orto2}, considered only numerically,  both sets of
arithmetic vectors $\left\{{\partial(\sigma\circ\varphi^{-1})\over\partial x^i}\right\}$,
$\left\{{\partial (x^i\circ\chi)\over
\partial u},{\partial (x^i\circ\chi)\over
\partial v}\right\}$ are mutually orthogonal with respect to the standard dot product. Since $g_{ik}$
works as the transition matrix between covariant and usual, contravariant  components of a vector
in $T_pM$,  the tangent vectors, say $m^i$, $q^i$, $r^i$, corresponding to the geodesics emanating
from $p$ to $m$, $q$, $r$, respectively, are linearly dependent. On the other hand, linearly
dependent $m^i$, $q^i$ would imply the points $p$, $q$, $m$ lying on the same geodesic, whose
segment would be contained in $\chi(I\times I)\subseteq C_i \cap\fr J^+(m)$. Then $m^i$, $q^i$
would be linear combinations of ${\partial (x^i\circ\chi)\over
\partial u}$ and ${\partial (x^i\circ\chi)\over \partial u}$, which then would imly a contradiction with
\ref{orto1} and \ref{orto2}. Beside this, it is not possible that $m=q$ since it would lead to
$\dim (C_i\cap \fr J^+(m))=3$.

Hence, $m^i$, $q^i$ are linearly independent and so there exist $a,b\in\R$, $(a,b)\ne (0,0)$, such
that $r^i=a m^i+ b q^i$. Further, from the fact that $m^i$, $q^i$, $r^i$ are (non-zero) null
vectors, it follows
\begin{multline}
0=g_{ik} r^i r^k=g_{ik} (a m^i+ b q^i)(a m^k+ b q^k)= \\ g_{ik} (a^2 m^i m^k + ab m^i q^k + ab m^k
q^i + b^2 q^i q^k)= \\ a^2 g_{ik}m^i m^k + 2 ab g_{ik} m^i q^k + b^2 g_{ik} q^i q^k = 2 ab g_{ik}
m^i q^k,
\end{multline}
which implies
\begin{equation}\label{metric1}
ab g_{ik} m^i q^k =0
\end{equation}

Now, let us choose the coordinate system $\varphi=(x^1, x^2, x^3, x^4)$ in such a way that the
non-diagonal elements of the metric $g_{ik}$ vanish  and $g_{11}=g_{22}=g_{33}=1$, $g_{44}=-1$ at
$p$. Then, for $a\ne 0\ne b$, the equation \ref{metric1} will have the form

\begin{equation}\label{vectors1}
\sum_{i=1}^3 m^i q^i=m^4 q^4
\end{equation}

and since $m^i$, $q^i$ are null vectors, also it holds

\begin{equation}\label{vectors2}
\sum_{i=1}^3 (m^i)^2=(m^4)^2, \quad  \sum_{i=1}^3 (q^i)^2=(q^4)^2.
\end{equation}

Then
\begin{equation}
\left(\sum_{i=1}^3 m^i q^i\right)^2=\left[\sum_{i=1}^3 (m^i)^2\right] \cdot \left[\sum_{i=1}^3
(q^i)^2\right],
\end{equation}
but this is possible only if there exists some $c\in\R$, $c\ne 0$, such that $q^i=c\, m^i$ for
$i=1,2,3$ (as it follows from the properties of the well-known Cauchy-Bunyakovsky-Schwarz
inequality). Further, \ref{vectors1} and \ref{vectors2}  imply
\begin{equation}
(q^4)^2=\sum_{i=1}^3 (q^i)^2=\sum_{i=1}^3 c\, m^i q^i= c\, \sum_{i=1}^3  m^i  q^i=c \, m^4\, q^4
\end{equation}
and since $q^4\ne 0$, it holds $q^4=c\, m^4$. But it means that the vectors $m^i$, $q^i$ are
collinear, which we already excluded. Hence, it holds $a=0$, $b\ne 0$ or $a\ne 0$, $b=0$.

The first case means that $r^i=b q^i$, so the points $q$, $r$ lie on the same null geodesic as $p$,
but $q=r$ is not possible, since by our previous considerations $C_i$, $C_j$ belong to the
boundaries of different cones $J^\pm(q)$, $J^\pm(r)$. Take any element $p^\prime\in C_i\cap C_j$.
Suppose, for certainty, that $r\leqslant q$, so the null geodesic segment emanating from $r$ and
terminating at $q$ is future-oriented. Then there exists a non-timelike path from $r$ to
$p^\prime$, say $\eta$, consisting of the null geodesic segment from $r$ to $q$ and the null
geodesic segment emanating from $q$ and terminating at $p^\prime$. Suppose that $\eta$ is not a
null geodesic path. Then, by Lemma \ref{global curves}, there exist a timelike curve $\zeta$
connecting $r$ with $p^\prime$. Since $p^\prime\in C_j$, the point $p^\prime$ also lies on a null
geodesic, say $\gamma$, emanating from $r$. But this is not possible on a normal convex
neighborhood $U=U_{F,G}$, containing $F\cup G$ and $F\Diamond G$ by Definition~\ref{multidiamond}.
Indeed, if so, then the timelike curve $\zeta$ could be replaced by the longest timelike curve,
connecting $p^\prime$ with $r$, which would be also a geodesic. Then there would exist two
different geodesics between $p^\prime$ and $r$, however, by the standard theory of differential
equations, it is not possible on $U$. Hence, $\eta=\gamma$ is the only geodesic path connecting
$p^\prime$ and $r$. But, that path contains totally the set $C_i\cap C_j$, so $\dim(C_i\cap
C_j)=1$, which contradicts to the assumption that $\dim(C_i\cap\fr J^+(m))=\dim \chi(I\times I)=2$.
The case of $q\leqslant r$ could be solved analogically, only by exchanging the roles of $q$, $r$.
Finally, it remains only the possibility that $a\ne 0$ and $b=0$, which means that the points $m$,
$r$ are on the same null geodesic as $p$, which implies that $r\leqslant m$ or $m\leqslant r$.

(iii$_1$) Using Lemma \ref{cones1},  we can similarly as in (ii$_1$)show that $E_j=\fr J^-(r)$,
$r\in G$. Then $A$ is contained in the geodesical segment between the points $m$, $r$. But this
contradicts to the main assumption of (iii), that $\dim (C_i\cap\fr J^+(m))=2$.

(iii$_2$) Hence, $C_j\subseteq E_j=\fr J^+(r)$. It is not possible that $r\leqslant m$ and $r\ne m$
since $r$ would not be a maximal element among those giving the same value for the set $A=F\Diamond
G$. Indeed, then $A\subseteq J^+(m)\subseteq J^+(r)$, which implies that $A=J^+(m)\cap A\subseteq
J^+(m)\cap A_r\subseteq J^+(r)\cap A_r=A$, where we denoted $A_r=(F\smallsetminus\{r\})\Diamond G$
(as in Lemma~\ref{maximize}). Then $r$ could be replaced by the greater element $m$, which is a
contradiction. Hence, it holds $m\leqslant r$. Then $A\subseteq J^+(r)\subseteq J^+(m)$ and from
the assumption of maximality of $m$ it follows that $m=r\in F$.
\end{proof}

\bigskip

Now, let $A,B\in\P$.   We put $A\prec B$ if $A\ne B$ and for every $a\in A$, $b\in B$, $a\leqslant
b$.

\begin{theorem}
$(\P, \subseteq,\prec)$ is a causal site.
\end{theorem}

\begin{proof}
First of all, we need to show that $\prec$ is a transitive binary relation on the set
$\P\smallsetminus \{\varnothing\}$ (the anti-reflexivity of $\prec$ follows directly from the
definition). Suppose that $A\prec B$ and $B\prec C$, where $A, B,C $ are non-empty elements of
$\P$. Let $a\in A$, $c\in C$. Since $B\ne\varnothing$, there is some $b\in B$. By definition,
$a\leqslant b\leqslant c$, and since $\leqslant$ is a transitive relation, we have $a\leqslant c$.
Suppose that $A=C$. Then $A\prec B$ and $B\prec A$. It means that $x\leqslant y$ and $y\leqslant x$
for all $x\in A$ and $y\in B$. Since $A\ne B$, it is possible to choose some concrete $x\ne y$.
Then it follows that there exist two future-oriented non-spacelike curves, where the first one
emanates from x and terminates at y, and the second curve vice-versa. Then $M$ admits of a closed
non-spacelike  curve, which is not possible. It means that $A\ne C$ and so the relation $\prec$ is
transitive.

The conditions (i)-(iii) of the definition of the causal site are only immediate consequences of
the properties of the set inclusion and the fact that $M$ does not admit of the closed
non-spacelike curves. We leave their verification  to the reader. It remains to check the last
axiom (iv).

\medskip

At first, suppose that $A$ consists of only a single multi-diamond. Denote
$$O_A=\{x|\, x\in D, A\subseteq J^+(x)\}.$$
Let $M_A$ be the set of all maximal elements of $O_A$ (with respect to the order $\leqslant $).
Choosing $p\notin F$ in Lemma~\ref{maximize} we can conclude that for every $x\in O_A$ there exists
$m\in M_A$ such that $x\leqslant m$. We put
$$A_\bot=\bigcup_{m\in M_A} J^-(m),$$ and for $B\in \P$, $B\ne A$ we denote
$$B_A=B\cap A_\bot.$$
By Lemma~\ref{maxfinite}, $M_A$ is finite, so $B_A\in \P$. Let $b\in B_A$, $a\in A$. By the
definition of $B_A$, there exists some $m\in M_A$ with $b\in J^-(m)$, so $b\leqslant m$. We also
have $a\in A\subseteq J^+(m)$, so $m\leqslant a$. Then $b\leqslant a$, which implies $B_A\prec A$.
Suppose that $C\prec A$, $C\subseteq B$ for some $C\in \P$. Let $c\in C$. If $a\in A$, then
$c\leqslant a$, which gives $a\in J^+(c)$. Therefore, $A\subseteq J^+(c)$. Then $c\in O_A$, so
there exists $m\in M_A$, such that $c\leqslant m$. Then $c\in J^-(m)\subseteq A_\bot$. Hence,
$C\subseteq A_\bot$, which together with $C\subseteq B$ gives the requested inclusion $C\subseteq
B_A$. Then $B_A$ is the cutting of $B$ by $A$.

\medskip

Now, consider the general case that $A=\bigcup_{i=1}^n A_i$, where $A_i\in\mathcal D$ are non-empty
multi-diamonds for $i\in\{1,2\dots, n\}$. For $B\in \P$ we put

$$B_A=\bigcap_{i=1}^nB_{A_i}.$$

\medskip

Since $B_{A_i}\subseteq B$, it is clear that $B_A\subseteq B$. Also we have $B_{A_i}\prec A_i$ for
every $i\in\{1,2\dots, n\}$. Take $b\in B_A$ and $a\in A$. It follows that $b\leqslant a$, since
there exists $i\in\{1,2\dots, n\}$ such that  $a\in A_i$, and $b\in B_A\subseteq B_{A_i}$,
$B_{A_i}\prec A_i$. Suppose that $B_A=A$. Then, for $i\in\{1,2\dots, n\}$, $A_i\subseteq
A=B_A\subseteq B_{A_i}$, which implies $A_i\subseteq B_{A_i}\prec A_i$. By the already verified
condition (i), it follows $A_i\prec A_i$. Since the multi-diamonds $A_i$ are non-empty, this is not
possible.

Now, consider $C\in \P$ such that $C\prec A$, $C\subseteq B$. Then, certainly, $C\prec A_i$ for
$i\in\{1,2\dots, n\}$ by (ii).  But then $C\subseteq B_{A_i}$ for $i\in\{1,2\dots, n\}$ and so
$C\subseteq \bigcap_{i=1}^n B_{A_i}=B_A$. Hence,  $B_A$ is the cutting of $B$ by $A$.
\end{proof}

Let $\pi$ be the family of all maximal centered subsets of $\P$.

\begin{theorem} The topological space $(X,\tau)$ corresponding to the framework $(\P^d,\pi^d)$ is homeomorphic to $M$ equipped with the co-compact topology.
\end{theorem}

\begin{proof}
 As we already defined before, $X=\P^d=\pi$. We will show that any point $p\in M$ defines a maximal centered subset of $\P$, say $f(p)=\{C|\,
C\in \P, p\in C\}$. The family $f(p)$ obviously is centered, since $\P$ is closed under finite
intersections and $f(p)$ contains those elements of $\P$, whose  contain $p$. Let $Q$ be another
centered family such that $f(p)\subseteq Q\subseteq \P$. Suppose that there is some $F\in Q$, such
that $p\notin F$. The set $M\smallsetminus F$ is an open neighborhood of $p$ with respect to the
manifold topology $\tau_M$. By Lemma~\ref{interiors}, there exists $A\in \P$ with $x\in \Int
A\subseteq A\subseteq M\smallsetminus F$. Then $A\cap F=\varnothing$, which contradicts to the
assumption that $Q$ is centered. Thus, all elements of $Q$ contain $p$, which means that $Q=f(p)$.
Now it is clear that $f(p)$ is a maximal centered subfamily of $\P$.

Conversely, let  $Q\in\pi$. By Lemma~\ref{cocompactbase}, the elements of $Q$ are closed with
respect to the co-compact topology on $M$. Since $M$ is compact with respect to its co-compact
topology, $\bigcap_{F\in Q}F\ne\varnothing$. Suppose that there exist $\{x,y\}\subseteq
\bigcap_{F\in Q}F$, where $x\ne y$. Since $M$ is a T$_1$ space, $M\smallsetminus\{y\}$ is an open
neighborhood of $x$. By Lemma~\ref{interiors}, there is some $B\in \P$ with $x\in \Int B\subseteq
B\subseteq M\smallsetminus \{y\}$. The collection $Q\cup\{B\}\subseteq \P$ is centered and since
$y\notin\ B$, it is an extension of $Q$. But this contradicts to the assumption that  $Q$ is
maximal. Thus, the intersection $\bigcap_{F\in Q}F$ contains only one element, say $g(Q)$.
Consequently we have $g(f(p))=p$ and $f(g(Q))=Q$. Hence, the mappings $f:M\map X$ and  $g:X\map M$
are bijections, inverse to each other.

Furthermore, for $A\in \P$ we have $g^{-1}(A)=\{Q|\, Q\in X, g(Q)\in A\}=\{Q|\, Q\in\pi, Q\in
f(A)\}=\{Q  |\, Q\in\pi, A\in Q \}=\pi(A)$, which is a subbasic closed set in $(X,\tau)$. Then
$g:X\map M$ is continuous.

Now, take a set $\pi(B)$, where $B\in \P$, from the closed base $\pi^d$ of $\tau$. Then
$f^{-1}(\pi(B))=\{p|\, p\in M, f(p)\in\pi(B)\}=\{p|\, p\in M, B\in f(p)\}$. For every $p\in
f^{-1}(\pi(B))$, $f(p)$ is a maximal centered subfamily of $\P$, containing the set $B$. As we have
shown above, its intersection contains the only element $g(f(p))=p$. So $f^{-1}(\pi(B))=\{p|\, p\in
M, p\in B\}=B$. Since $B$ is a compact set with respect to the manifold topology $\tau_M$ on $M$,
it is closed with respect to the co-compact topology on $M$ and so the map $f: M\map X$ is
continuous. Hence, the spaces $(X,\tau)$ and $M$, equipped with the co-compact topology, are
homeomorphic.
\end{proof}

\bigskip

\section{Finite approximations}\label{approximations}

\medskip

It is well-known that Albert Einstein asked the question if the Universe is finite in the sense of
volume, mass and energy (see, for example, \cite{Ei}). Whatever it is the correct answer, all our
experience with the Universe is finite, since in a finite time one can do only a finite number of
observations, measurements or experiments. Hence, everything what we know about the Universe, is a
result of extrapolation of that our finite experience. Thus, it may have some sense to study which
framework or topological structures may arise from a process of forming and completing of a family
of finite frameworks, representing our growing, but still finite experience with the Universe.

\medskip

\begin{definition}
 Let $(X,\alpha)$ be a framework, $Y\subseteq X$. Denote $\beta=\{U\cap Y|\, U\in\alpha\}$. Then $(Y,\beta)$ is called the {\it
induced subframework} of $(X,\alpha)$.
\end{definition}

 We put $\pi_K=\{U\cap K|\, U\in\pi \}$ for every finite $K\subseteq \P$.
Ovbiously, if $K$, $L$ are finite subsets of $\P$ and $K\subseteq L$, $(K,\pi_K)$ is an induced
subframework of $(L,\pi_L)$ and both are induced subframeworks of the original framework
$(\P,\pi)$. The collection of finite frameworks $(\P_K, \pi_K)$ is directed by the set inclusion.
Let
$$\sigma=\{W|\, W\subseteq \P, W\cap K \in \pi_K\text{ for every finite } K\subseteq \P\}.$$
Obviously, $\pi\subseteq \sigma$. Moreover, after a restriction to a finite family $K\subseteq \P$
of places in the framework $(\P,\pi)$ there is no way how to distinguish between $(\P,\pi)$  and
$(\P,\sigma)$, since $$\{U\cap K|\, U\in\pi\}=\pi_K=\{W\cap K|\, W\in\sigma\}.$$ It could seem that
it would be a good idea to approximate $(\P,\pi)$ by $(\P,\sigma)$. However, as we will show later,
$(\P,\sigma)$ may contain too many abstract points (that is, elements of $\sigma$) in comparison to
$(\P,\pi)$.

Let $\lambda\subseteq\sigma$ be a chain linearly ordered by the set inclusion. We put
$L=\bigcup\lambda$. Clearly, $L\subseteq \P$. If $K\subseteq \P$ is finite, then also the set
$L\cap K=\bigcup_{W\in\lambda} (W\cap K)$ is finite. Denote $L\cap K=\{x_1, x_2, \dots, x_k\}$.
Then for every $i\in\{1,2,\dots,k\}$, there is some $W_i\in\lambda$ with $x_i\in W_i$. But
$\lambda$ is a chain, so there is the greatest element, say $W_m\in\sigma$, among  all $W_1,
W_2,\dots, W_k$ with respect to $\subseteq$. Then $L\cap K=W_m\cap K\in\pi_K$. Thus, $L\in\sigma$,
so $L$ is the upper bound of the chain $\lambda$. By Zorn's Lemma, every element $W\in\sigma$ is
contained in some maximal element $M\in\sigma$. Let $\mu\subseteq\sigma$ be the set of all maximal
elements of $\sigma$. The framework $(\P,\mu)$ could be another candidate for an approximation of
$(\P,\pi)$.

% However, in general case it easily may happen that $\pi\cap\mu=\varnothing$, so $(\P,\mu)$ could
% have completely different abstract points than the original framework $(\P,\pi)$.

\begin{example}
 Let $\P=\N$ and let $\pi$ be the set of all finite subsets of $\P$. Then, respecting the previous denotations, $\sigma=2^\P$, and
$\mu=\{\P\}$.
\end{example}

\begin{proof}
It is obvious, that $\sigma\subseteq 2^\P$. Let $W\in 2^\P$. For every finite $K\subseteq \P$,
$W\cap K\in\pi_K=2^K$, so $W\in\sigma$. However, the set $\sigma=2^\P$ has only one maximal element
with respect to the set inclusion, $\P$.
\end{proof}

The following theorem now describes the approximation properties of our construction under very
general topological conditions.

\begin{theorem}
 Let $(X,\tau)$ be a topological T$_1$ space, $\C$ the family of all closed
sets. Let $(\P,\pi)$ be the dual framework of $(X,\C)$. Then the dual of $(\P,\mu)$ generates the
Wallman compactification of $(X,\tau)$. More precisely, $\mu^d$ is a closed subbase of $\omega X$.
\end{theorem}

\begin{proof}
We have $\P=\C$ and $\pi=\{\C(x)|\, x\in X\}$, where $\C(x)=\{C|\, C\in\C, x\in C\}$. We will show
that every element of $\sigma$ is a family of closed sets of the topological space $(X,\tau)$,
having f.i.p. Let $W\in\sigma$ and let $K\subseteq W$ be finite. Then $K=W\cap K\in\pi_K$, so there
exists $y\in X$ such that $K=\C(y)\cap K$. Then $K\subseteq\C(y)$, which gives $y\in\bigcap
K\ne\varnothing$. Hence, $W$ has f.i.p and its closedness follows from the fact that $W\subseteq
\P=\C$.

Let us show that $\pi\subseteq\mu$. Suppose that for some $W\in\sigma$ we have $\C(x)\subseteq W$.
Since $(X,\tau)$ is a  T$_1$ space, $\{x\}\in\C(x)\subseteq W$. But $W$ has f.i.p, so every its
element must contain $x$. Then $W\in\C(x)$, so $W=\C(x)$. Therefore, $\C(x)$ is a maximal element
of $\sigma$, that is, $\pi\subseteq \mu$.

Let $\eta$ be the family of all maximal collections of closed sets having f.i.p. Note that, in
other words, $\eta$ is the family of all ultra-closed filters on $(X,\tau)$. We will show that
$\eta=\mu$. As the first step, we will prove that $\eta\subseteq\mu$. Let $U\in\eta$. Take any
finite $K\subseteq \P=\C$ and denote $L=U\cap K$. The set $L$ contains only finitely many elements
of $U$. The family $U$ has f.i.p., so $\bigcap L\ne\varnothing$. Denote $D=\bigcup (K\smallsetminus
L)$. The set $D$ is closed (and could be possibly empty, if $K=L$). Suppose that $\bigcap L=\bigcap
(L\cup\{D\})$. Consider the family $U\cup\{D\}$. If $F\subseteq U\cup\{D\}$ is finite, then
$F\smallsetminus\{D\}$ and also $(F\smallsetminus\{D\})\cup L$ are  finite subsets of $U$, so
$\varnothing\ne\bigcap((F\smallsetminus \{D\})\cup L)=\bigcap(F\smallsetminus \{D\})\cap (\bigcap
L)=(\bigcap(F\smallsetminus \{D\}))\cap (\bigcap(L\cup\{D\}))=(\bigcap (F\smallsetminus\{D\}))\cap
D\cap (\bigcap L)=(\bigcap F)\cap (\bigcap L)\subseteq \bigcap F$. Then $U\cup\{D\}$ has f.i.p. In
particular,  $D\ne\varnothing$, which implies that $K\ne L$. Then $K\smallsetminus
L\ne\varnothing$, and $U\cap(K\smallsetminus L)=\varnothing$. It follows from the maximality of $U$
that $D=\bigcup(K\smallsetminus L)\notin U$. Then $U\cup\{D\}$ is a strictly greater family than
$U$. This contradicts to the maximality of $U$. Therefore, there is some $z\in X$ such that
$z\in\bigcap L$, $z\notin D$. Then $L\subseteq\C(z)$, but $(K\smallsetminus L)\cap\C(z)=
\varnothing$. That means $U\cap K= L= L\cap\C(z)=K\cap\C(z)\in\pi_K$. Then $U\in\sigma$. By
definition of the set $\mu$, there exists $W\in\mu$ such that $U\subseteq W$. But as we have shown
above, $W$ is a family of closed sets having f.i.p. By maximality of $U$, we have $U=W$, so
$U\in\mu$. Therefore, $\eta\subseteq\mu$.

Conversely, let $U\in\mu$. Because also $U\in\sigma$, the family $\sigma$ consists of closed sets
and has f.i.p. Then there exists some $W\in\eta$ with $U\subseteq W$. By the previous paragraph, we
have $W\in\mu\subseteq\sigma$. But $U$ is a maximal element of $\sigma$, so $U=W$ and  $U\in\eta$.
Together we finally have $\mu=\eta$.

Now, consider the framework $(\P^d, \mu^d)$. It holds $\P^d=\mu$, $\mu^d=\{\mu(C)|\, C\in\C\}$,
where $\mu(C)=\{U|U\in\mu, C\in U\}$. Consider the topological space $(Y,\vartheta)$, where
$Y=\P^d$ and its topology is generated by its closed subbase $\mu^d$. Consider $\Psi\subseteq\C$,
such that for every $C_1, C_2,\dots, C_k\in \Psi$ it holds
$$\mu(C_1)\cap\mu(C_2)\cap\dots\cap\mu(C_k)\ne\varnothing. $$ There exist $U\in\mu$ (depending on
the selection of $C_1$, $C_2$, \dots $C_k$), such that $C_1,C_2,\dots C_k\in U$, so
$\varnothing\ne\bigcap_{i=1}^k C_i\in U$. Then $\Psi$ has f.i.p., so there exists a maximal family
$W\subseteq\C=\P$, having f.i.p. and containing $\Psi$.  By the previous paragraph, $W\in\mu$. Now,
if $C\in\Psi$, then also $C\in W$, which gives $W\in\mu(C)$ and so
$$W\in\bigcap_{C\in\Psi}\mu(C)\ne\varnothing.$$ Therefore, $(Y,\vartheta)$ is compact.

\medskip

Finally, consider the mapping $f:X\map Y$ , where $f(x)=\C(x)$. Clearly, $f$ is an injection.
Indeed, for $x\ne y$  we have $\{x\}\in\C(x)$ but $\{x\}\notin\C(y)$, so $f(x)\ne f(y)$. Let
$C\in\C$. Then $f^{-1}(\mu(C))=\{x|\, x\in X, f(x)\in\mu(C)\}=\{x|\, x\in X,
\C(x)\in\mu(C)\}=\{x|\, x\in X, C\in\C(x)\}=\{x|\, x\in X, x\in C\}=C$, so $f$ is continuous.
Further, for any $D\in\C$, $f(D)=\{\C(x)|\, x\in D\}=\{\C(x)|\, x\in X,
\C(x)\in\mu(D)\}=f(X)\cap\mu(D)$, so f is also a closed mapping. Then $f$ is a homeomorphous
embedding of $(X,\tau)$ to the compact space $(Y,\vartheta)$. Moreover, $f(X)=\pi$, so the elements
of $X$ and the families $\C(x)$, which constitute the principal ultra-closed filters generated by
the elements of $X$, may be identified. Consider the set $\mu\smallsetminus \pi$. Then every its
element $W\in \mu\smallsetminus \pi$ is vanishing (that is, non-convergent) -- otherwise, because
of maximality,  $W=\C(z)$, where $z$ is the unique element from $\bigcap W$, which would imply
$W\in \pi$. But then $Y$ is the underlying set of the Wallman compactification of $(X,\tau)$ and
$\mu^d=\{\mu(C)|\, C\in\C\}$ is its closed base.
\end{proof}

\medskip

Among others the previous theorem means that for a compact $T_1$ topological space, its
approximation by a suitable family of finite frameworks may achieve arbitrary precision.

\bigskip

\section{Some Final Remarks in Historical Context}

\medskip

The progress in mathematical and theoretical physics witnesses that various applications of
topology in physics may be far-reaching and illuminating. It could be very difficult to track down
the origins of such applications, but one of the first attempts may be associated with the year
1914, when A. A. Robb came with his axiomatic system for  Minkowski space $\M$, analogous to the
well-known axioms of Euclidean plane geometry. In \cite{Rb} he essentially proved that the
geometrical and topological structure of $\M$ can be reconstructed from the underlying set and a
certain order relation among its points. As it is noted in \cite{Do},  some prominent
mathematicians and physicists criticized the use of locally Euclidean  topology in mathematical
models of the space-time. Perhaps as a reflection of these discussions, approximately at the same
time when de Groot wrote his papers on co-compactness duality, there appeared two interesting
papers \cite{Ze} and \cite{Ze2}, in which E. C. Zeeman studied an alternative topology for
Minkowski space. The Zeeman topology, also referred as the fine topology, is the finest topology on
Minkowski space, which induces the $3$-dimensional Euclidean topology on every space-axis and the
1-dimensional Euclidean topology on the time-axis. Among other interesting properties, it induces
the discrete topology on every light ray. A. Kartsaklis in \cite{Ka} studied connections between
topology and causality. He attempted to axiomatize causality relationships on a point set, equipped
with three binary relations, satisfying certain axioms, by a structure called a {\it causal space}.
He also introduced so called {\it chronological topology}, the coarsest topology, in which every
non-empty intersection of the chronological future and the chronological past of two distinct
points of a causal space is open.

\medskip

In the camp of quantum gravity, there appeared similar efforts and attempts to get some gain from
studying the underlying  structure of space-time -- topological, geometrical or discrete --
however, significantly later. The possible motivation is explained, for instance, in \cite{Ro} by
C.~Rovelli. It is pointed out that the loop quantum gravity leads to a view of the space-time
geometry extremely different from the smooth model at the shortest scale level. Also the topology
of space-time at Planck scales could be very different from that we meet in our everyday experience
and which has been originally extrapolated from the fundamental concepts of the continuous and
differentiable mathematics. Thus, the usual properties and attributes of the space-time, like its
Hausdorffness or metrizability may not be satisfied (for a groundbreaking paper, see \cite{HPS}).
The most important source of inspiration for our paper was the work \cite{CC} of  J. D. Christensen
and L. Crane. Motivated by certain requirements of their research in quantum gravity, these authors
developed a novel axiomatic system for the generalized space-time, called {\it causal site},
qualitatively different from the previous, similar attempts. The notion itself is a successful
synthesis of two other notions, a  Grothendieck site (which is a small category equipped with the
Groethen\-dieck topology) \cite{Ar} and a  causal set of R. Sorkin \cite{So}. One of the most
important merits of the new axiomatic system it is that the causal site is a pointless structure,
not unlike to certain well-known concepts of pointless topology and locale theory. We should also
mention the work of K. Martin and P. Panangaden who studied causality with help of the theoretical
background of domain theory (see, for example, \cite{MP1} and \cite{MP2} and our note in
Section~\ref{causal} after Corollary~\ref{NonContinuous}).

\bigskip

We close the paper by a summary of the main ideas and results:

\medskip

\begin{enumerate}

\item Every causal site generates, by a canonical way, an associated compact $T_1$ topology (and,
naturally, also its de Groot dual). We call this topology {\it weakly causal topology}.

\item The axioms of a causal site connect both of the relations $\sqsubseteq$, $\prec$ together to
one organic unit, so they are no longer independent. However, the strengths of their connection may
vary. Some candidates for a `causal' relation $\prec$ admits of more different appropriate
`inclusion' relations $\sqsubseteq$, and for some of them any corresponding `inclusion relation may
not exist.

\item For every globally hyperbolic Lorentzian manifold there exists a causal site, whose weakly
causal topology is the de Groot dual of the manifold topology and vice versa, that is, the manifold
topology is also the de Groot dual of the weakly causal topology. Thus, the weakly causal topology
and so also the causal site itself has the full information about the usual, manifold topology.

\item The weakly causal topology is weaker than the manifold topology, but in the sense of
mathematical analysis, it is capable of doing the same job as the manifold topology. Both of the
topologies coincide on compact subsets, and so on finite distances, so there is no physical way how
to distinguish between them. Both of these topologies are possible extrapolations of the human's
finite experience with the universe, but because it seems that nature follows a certain `principle
of minimality' in the natural laws, the weakly causal topology may be considered more natural or
fundamental.

As it had been remarked by de Groot in \cite{Gro}, from the philosophical point of view, the
co-compact topology is naturally related to the notion of potential infinity -- in a contrast to
the notion of actual infinity, which is mostly used in the traditional mathematical approach. To
illustrate the difference, consider a countably infinite sequence $x_1, x_2, \dots$ of points lying
on a straight line in space or space-time, with the constant distance between $x_i$ and its
successor $x_{i+1}$. In the usual, Euclidean topology, the sequence is divergent and it approaches
to an improper point at infinity. To make it convergent, one need to embed the space into its
compactification (for instance, the Alexandroff one-point compactification is a suitable one). The
points completed by the compactification then appear at the infinite distance from any other point
of the space. On the other hand, the co-compact topology, which locally coincides with the usual
topology, is already compact and superconnected, so the sequence $x_1, x_2, \dots$ lies eventually
in each neighborhood of every point.

\item In this sense, the universe is `co-compact': the usual, manifold topology is co-compact with
respect to the weakly causal topology and vice versa.

\item There exist causal sites, which cannot be generated by any Lorentzian manifold, and whose
causal relation or its reflexive closure is not a continuous poset. These causal sites still
generate their weakly causal topology and its de Groot dual, which is a corresponding counterpart
of the manifold topology of the relativistic space-time. In such a case, however, the weakly causal
topology need not satisfy the identity $\tau=\tau^{GG}$ and so up to four topologies may arise as
the de Groot dual iterations of the weakly causal topology.

\item Although it is (so far) unknown if non-Lorentzian causal sites have any physical sense, they
could be potentially used for investigation of various physical relationships on the particle
level, beyond the smooth space-time models, as in quantum gravity, in very short distances (below
the Planck length) or in cases in which the metric properties of the space-time collapse to
one-point singularity (and its topological properties still may have some continuation).

\item In Section~\ref{howto} we developed a general method suitable for topologizing various
situations, including physics and computer science. The method is also useful for working with
various scales using different and 'incompatible' topological  models ('smooth' vs. 'discrete') if
considered in the traditional point-set approach.

\item In Section~\ref{approximations} we proved that a compact T$_1$ topological space can be
expressed by a limit of finite framework approximations. Since the weakly causal topology,
naturally generated by every causal site, is compact T$_1$, this topology can be reached by the
process of finite approximations as a limit. According to Section~\ref{Lorentz}, the co-compact
topology on Lorentzian globally hyperbolic manifolds, for example, can be reached in this way.

\item It could be proven that open diamonds also form a causal site, which generates the manifold
topology of a globally hyperbolic Lorentzian manifold directly, without the necessity of using the
de Groot dual. However, this approach would give much weaker results. The topology arising from
regions used as a subbase for open sets does not have, in general, any interesting or significant
properties. In particular, it need not be compact T$_1$ even for a special causal site, generated
by a Lorentzian, globally hyperbolic manifold. As a consequence, the topology need not be
approximable by the family of finite frameworks, described in Section~\ref{approximations}. This
also indicates, that the perceived topology, generated from our extrapolated finite experience with
causal relationships in the Universe, is rather the co-compact topology, than its de Groot dual,
the usual, locally Euclidean manifold topology.

\end{enumerate}

\bigskip

\bigskip

\bibliographystyle{amsplain}

\end{document}